\documentclass[journal]{IEEEtran}

\usepackage{epsfig} 
\usepackage{mathptmx} 
\usepackage{times} 
\usepackage{pgfplots}

\usepackage{amssymb,latexsym,amsfonts,amsmath,amsthm}
\usepackage{graphicx}
\usepackage{color}
\usepackage{float}

\usepackage{psfrag}
\usepackage{pstricks}
\usepackage{pstricks-add}
\usepackage{graphics} 
\usepackage{pgfplots}
\usepackage{balance}

\usepackage[ruled]{algorithm}
\usepackage{algpseudocode}
\newcommand{\algorithmicinput}{\textbf{Input:}}
\newcommand{\algorithmicoutput}{\textbf{Output:}}
\def\Input{\item[\algorithmicinput]}
\def\Output{\item[\algorithmicoutput]}

\def\bO{\bar{\Omega}}

\def\bH{\bar{H}}
\def\bh{\bar{h}}
\def\bG{\bar{G}}
\def\bg{\bar{g}}
\def\bx{\bar{x}}
\def\ee{\hspace{-0.1cm} & \hspace{-0.1cm} }

\everymath{\displaystyle}

\newtheorem {theorem}{Theorem}
\newtheorem {lemma}{Lemma}

\newtheorem {definition}{Definition}

\newtheorem {example}{Example}
\newtheorem {assumption}{Assumption}
\newtheorem {remark}{Remark}
\newtheorem {proposition}{Proposition}

\def\R{\mathbb{R}}
\def\B{\mathbf{B}}
\def\N{\mathbb{N}}

\def\B{\mathcal{B}}

\def\co{\textrm{co}}

\def\int{\text{int} }
\def\proj{\text{proj}}

\title{\LARGE \bf Computing control invariant sets is easy}

\author{Mirko~Fiacchini, Mazen~Alamir\thanks{M. Fiacchini and M. Alamir are  
with Univ. Grenoble Alpes, CNRS, Gipsa-lab, F-38000 Grenoble, France. \newline
{\tt \{mirko.fiacchini,mazen.alamir\}@gipsa-lab.fr}\newline
This work has been submitted to the IEEE for possible publication. Copyright 
may be transferred without notice, after which this version may no longer be 
accessible.}
}

\begin{document}
\maketitle

\begin{abstract}
In this paper we consider the problem of computing control invariant sets for 
linear controlled systems with constraints on the input and on the states. We 
focus in particular on the complexity of the computation of the N-step 
operator, given by the Minkowski addition of sets, that is the basis of many of 
the iterative procedures for obtaining control invariant sets. Set inclusions 
conditions for control invariance are presented that involve the N-step sets and 
are posed in form of linear programming problems. Such conditions are employed 
in algorithms based on LP problems that allow to overcome the complexity 
limitation inherent to the set addition and can be applied also to high 
dimensional systems. The efficiency and scalability of the method are 
illustrated by computing in less than two seconds an approximation of the 
maximal control invariant set, based on the 15-step operator, for a system whose 
state and input dimensions are 20 and 10 respectively.
\end{abstract}

\section{Introduction}

Invariance and contractivity of sets are central properties in modern control 
theory. For a dynamical system, a set is invariant if the trajectories starting 
within the set remain in it. For controlled systems, if the state can be 
maintained by an admissible input in the set, then it is referred as control 
invariant. Although the first important results on invariance date back to the 
beginning of the seventies \cite{BertsekasTAC72}, this topic gained 
considerable interest in the recent years, mainly due to its relation with 
constrained control and popular optimization-based control techniques as Model 
Predictive Control. The existence of an invariant set to be imposed as terminal 
constraint is, in fact, an essential ingredient to assure recursive feasibility 
and constraints satisfaction for many classical MPC control strategies 
\cite{MayneAUT00} as well as more recent techniques \cite{AlamirAUT17}.

The study of invariance and set theory methods for control gained interest also 
thanks to the foundational works by Blanchini and coauthors 
\cite{BlanchiniTAC94,BlanchiniAUT99,BlanchiniBook}. In these works, results are 
provided that proves that the existence of polyhedral Lyapunov functions, and 
then of contractive polytopes, are necessary and sufficient for stability of 
parametric uncertain linear systems \cite{MolchanovSCL89,BlanchiniAUT95}. 
Moreover, iterative procedures are given for the computation of control 
invariant sets that permit their practical implementation. Most of those 
procedures are substantially based on the one-step backward operator that 
associates to any set the states that can be steered in it by an admissible 
input, for every possible realizations of the eventual uncertainty. Different 
algorithms based on the one-step operator exist for computing control 
invariants, that substantially differs from the initial set. For instance, if 
the algorithm are initialized with the state constraints set, 
\cite{BlanchiniTAC94,KerriganPHD01,RunggerTAC17}, the one-step operator 
generates a sequence of outer approximations of the maximal control invariant 
that converges to it under compactness assumptions, see \cite{BertsekasTAC72}. 
Nevertheless, the finite determination of the algorithm, that is ensured for 
autonomous systems \cite{Kolm98}, cannot be assured in general in presence of 
control input. If, instead, the procedure is initialized with a control 
invariant set, a non-decreasing sequence of control invariant sets are obtained 
that converges from the inside to the maximal control invariant set, see the 
considerations on minimum-time ultimate boundedness problem in 
\cite{BlanchiniBook}. A particular case of the latter approach, that needs no 
preliminary knowledge of a control invariant set, suggests to initialize the 
procedure with the set containing the origin only (which is a control invariant 
in the general framework), obtaining the sequence of $i$-step null-controllable 
sets, that are control invariant and converges to the maximal control invariant 
set, see \cite{GutTAC86,DarupCDC14}. 

Thus, although the abstract iterative procedures for obtaining control invariant 
sets apply also for nonlinear systems, and some constructive results are given 
\cite{FiacchiniAUT10,FiacchiniSCL12}, the practical computation of the one-step 
set, that is the basis for them, is often prohibitively complex for their 
application in high dimension even in the linear context. A common solution to 
circumvent this major practical issue has been fixing the sets complexity to 
get conservative but more computationally affordable results. For instance, by
considering linear feedback and ellipsoidal control invariant sets, see the 
monograph \cite{Boyd94}, or by fixing the polyhedral set complexity 
\cite{BlancoIJC10,AthanasopoulosIJC14,TahirTAC15}. 

In this paper we address the main problem related to the complexity of the 
N-step operator, for discrete-time deterministic controlled systems, with 
polyhedral constraints on the input and on the state. Considering polyhedral 
sets, such operator can be expressed in terms of Minkowski sum of polyhedra and 
then as an NP-complete problem \cite{Tiwary08Hardness}, hardly manageable in 
high dimension. An algorithm is presented for determining control invariant sets 
that is based on a set inclusion condition involving the N-step set of a 
polyhedron but does not require to explicitly compute the Minkowski sum. Such 
condition is posed as an LP feasibility problem, then solvable even in high 
dimension. Once the condition is satisfied, the control invariant set is 
given by the convex hull of several k-step sets that can be represented through 
a set of linear equalities and inequality. A second algorithm, based on the 
previous results on Minkowski sum and convex hull, is also given. The methods, 
consisting in solving LP problems, are proved to be applicable to high 
dimensional systems. Examples that show the low conservatism and the high 
scalability of the approach are provided.

\paragraph*{Notations}
Denote with $\R_+$ the set of nonnegative real numbers. Given $n \in \N$, 
define $\N_n = \{x \in \N : 1 \leq x \leq n \}$. The $i$-th element of a 
finite set of matrices or vectors is denoted as $A_i$. Using the notation from 
\cite{RockafellarVariational}, given a mapping $M : \R^n \rightrightarrows 
\R^m$, its inverse mapping is denoted $M^{-1} : \R^n \rightrightarrows \R^m$. If 
$M$ is a single-valued linear mapping, we also denote, with slight abuse of 
notation, the related matrices $M \in \R^{n \times m}$ and, if $M$ is 
invertible, $M^{-1} \in \R^{m \times n}$. Given $a \in \R^n$ and $b \in \R^m$ 
we use the notation $(a, b) = [a^T \ b^T]^T \in \R^{n+m}$. The symbol $0$ 
denotes, besides the zero, also the matrices of appropriate dimensions whose 
entries are zeros and the origin of a vectorial space, its meaning being 
determined by the context. The symbol $\mathbf{1}$ denotes the vector of entries 
$1$ and $I$ the identity matrix, their dimension is determined by the context. 
The subset of $\R^n$ containing the origin only is $\{0\}$. The symbol $\oplus$ 
denotes the Minkowski set addition, i.e. given $C, D \subseteq \R^n$ then $C 
\oplus D = \{x + y \in \R^n: \ x \in C, \ y \in D\}$. To simplify the notation, 
the propositions involving the existential quantifier in the definition of sets 
are left implicit, e.g. $\{x \in A: f(x,y) \leq 0, \ y \in B \}$ means $\{x \in 
A: \ \exists y \in B \ \mathrm{ s.t. } f(x,y) \leq 0\}$. The unit box in $\R^n$ 
is denoted $\B^n$.

\section{Problem formulation and preliminary results}

The objective of this paper is to provide a constructive method to compute a 
control invariant set for controlled linear systems with constraints on the 
input and on the state. We would like to obtain a polytopic invariant set that 
could be computed through convex optimization problems. The main aim is to 
provide a method to obtain admissible control invariant sets for 
high-dimensional systems, thus no complex computational operations are supposed 
to be allowed.

The system is given by 
\begin{equation}\label{eq:system}
 x^+ = A x + B u
\end{equation}
with constraints
\begin{equation}\label{eq:XU}
x \in X = \{y \in \R^n : \ F y\leq f\}, \quad u \in U = \{v \in \R^m : \ G v 
\leq g\}.
\end{equation}

\begin{assumption}\label{ass:XU}
The sets $X$ and $U$ are closed, convex and contain the origin. 
\end{assumption}

Note that Assumption \ref{ass:XU} implies $f \geq 0$ and $g \geq 0$. Most of 
the iterative methods for obtaining invariant sets involve the image and 
the preimage of linear mappings.

\begin{remark}\label{rem:setalgebra1}
Given a polyhedron $\Omega = \{x \in \R^n: \ H x \leq h\}$, its preimage 
through the linear single-valued mapping $A : \R^n \rightrightarrows \R^n$, 
denoted $A^{-1} \Omega$, is well defined, even if matrix $A$ is singular. 
Indeed, $A^{-1} \Omega$ is the set of $x \in \R^n$ such that $A x \in \Omega$ 
and then it is given by 
\begin{equation*}\label{eq:A-1OmegaZ}
A^{-1} \Omega = \{x \in \R^n : A x \in \Omega\} = \{x \in \R^n: \ H A x \leq 
h\}, 
\end{equation*}
while the image of $\Omega$ through $A$ is 
\begin{equation*}\label{eq:AOmegaZ}
A \Omega = \{A x \in \R^n : x \in \Omega\} = \{A x \in \R^n: \ H x \leq h\}.
\end{equation*}
Moreover, for every $\gamma \in \R$ one has 
\begin{equation*}\label{eq:gammaOmegaZ}
\begin{array}{l}
\gamma \Omega = \{\gamma x \in \R^n : x \in \Omega\} = \{\gamma x \in \R^n: \ H 
x \leq h\}\\
\end{array}
\end{equation*}
and, defining the mapping $M : \R^n \rightrightarrows \R^n$ through the matrix 
$M = \gamma I$, both the image and the preimage of $\Omega$ through $M$ are 
defined. That is $M \Omega = \gamma \Omega$ and
\begin{equation*}\label{eq:invgammaOmegaZ}
\begin{array}{l}
M^{-1} \Omega = \{x \in \R^n : \gamma x \in \Omega\} = \{x \in \R^n: \ 
\gamma H x \leq h\}.
\end{array}
\end{equation*}
Note that, also in this case, $M^{-1} \Omega$ is well defined even for $\gamma 
= 0$: $M^{-1} \Omega = \R^n$ if $0 \in \Omega$ and $M^{-1} \Omega = \emptyset$ 
if $0 \notin \Omega$.
\end{remark}

The one-step backward operator is defined as 
\begin{equation*}
\begin{array}{l}
Q(\Omega) = A^{-1} (\Omega \oplus (-B U)) = \{x \in \R^n : Ax = y - Bu, \ u \in 
U, \ \\
\hspace{1.2cm} y \in \Omega \} = \{x \in \R^n : Ax + Bu \in \Omega, \ u \in U\}
\end{array}
\end{equation*}
and provides the set of points in the state space that can be mapped into 
$\Omega$ by an admissible input with dynamics (\ref{eq:system}). 

Considering $X = \R^n$, one way to obtain a control invariant set is by 
iterating the one-step operator starting from a given initial set $\Omega$, 
compact, convex set containing the origin in its interior, and then checking 
whether the union of the sets obtained at iteration $k$ contains $\Omega$. Thus 
the sketch of the algorithm is:

\begin{algorithm}[H]
\caption{Control invariant}\label{alg:I}
\begin{algorithmic}[1]
\Input matrices $A, B$, sets $\Omega$, $U$
\State $\Omega_0 \gets \Omega$
\State $k \gets 0$
\Repeat 
\State $\Omega_{k+1} \gets A^{-1} (\Omega_k \oplus (- BU))$
\State $k \gets k+1$
\Until {$\displaystyle \Omega \subseteq \co \left(\bigcup_{j = 1}^{k} 
\Omega_j\right)$}
\State $N \gets k$
\Output $\displaystyle \Omega_{\infty} \gets \co \left(\bigcup_{k = 1}^{N} 
\Omega_k \right)$
\end{algorithmic}
\end{algorithm}

In practice, a bound on the maximal number of iteration should be imposed to 
avoid an infinite loop. Considering the alternative, direct, definition of 
$\Omega_k$
\begin{equation}\label{eq:Ok}
\begin{array}{l}
\Omega_k = A^{-k} \left( \Omega \oplus \bigoplus_{i = 0}^{k-1} (-A^i B U ) 
\right) \\
= \{x \in \R^n : \ A^kx + \sum_{i = 0}^{k-1} A^i B u_{i+1} \in \Omega, \ u_i 
\in U \ \forall i \in \N_k\},
\end{array}
\end{equation}
the algorithm above reduces to search, given $\Omega$, for the minimal $N$ such 
that 
\begin{equation}\label{eq:invariance_cond1}
\Omega \subseteq \co \left(\bigcup_{k = 
1}^{N} \Omega_k\right)= \co \left(\bigcup_{k = 1}^{N} A^{-k} \left( \Omega 
\oplus \bigoplus_{i = 0}^{k-1} (-A^i B U ) \right)\right). 
\end{equation}
As a matter of fact, all the $N$ for which (\ref{eq:invariance_cond1}) holds, 
lead to a control invariant set. Moreover, if (\ref{eq:invariance_cond1}) is 
satisfied, then it is satisfied for every $K \geq N$, leading to a 
non-decreasing sequence of nested control invariant sets.

Thus, the algorithm computes the preimages of $\Omega$ until the stop condition 
(\ref{eq:invariance_cond1}) holds. Then all the states in $\Omega_{\infty}$ 
defined 
\begin{equation}\label{eq:Omegainfty}
\Omega_{\infty} = \co \left(\bigcup_{k = 1}^{N} \Omega_k \right)
\end{equation}
can be steered in $\Omega$, thus in $\Omega_{\infty}$ itself, in $N$ steps at 
most, by means of admissible controls, as proved in the following proposition. 

\begin{proposition}
Given $\Omega$ and $\Omega_j$ as defined in (\ref{eq:Ok}) if condition 
(\ref{eq:invariance_cond1}) holds for $k \in \N$ then the set $\Omega_\infty$ 
defined in (\ref{eq:Omegainfty}) is control invariant for the system 
(\ref{eq:system}) under the constraint $u \in U$.
\end{proposition}

\begin{proof}
Given $x \in \Omega_{\infty}$ we prove that there exists $u \in U$ such that 
$Ax+Bu \in \Omega_{\infty}$. From the definition (\ref{eq:Omegainfty}) of 
$\Omega_{\infty}$, $x \in \Omega_{\infty}$ implies the existence of $x_k \in 
\Omega_k$ and $\lambda_k \geq 0$, with $k \in \N_N$, such that $x = \textstyle 
\sum_{k = 1}^N \lambda_k x_k$ and $\textstyle \sum_{k = 1}^N \lambda_k = 1$. 
Moreover, by definition of $\Omega_k$, for every $y \in \Omega_k$ there exists 
$u_k(y) \in U$ such that $Ay+Bu_k(y) \in \Omega_{k-1}$, for all $k \in \N_N$ 
(and with $\Omega_0 = \Omega$). Then denoting $u_k = u_k(x_k)$ and defining 
$u(x) = \textstyle \sum_{k = 1}^N \lambda_k u_k$, one has that $u(x) \in U$ from 
convexity of $U$, and 
\begin{equation*}
\begin{array}{l}
A x + Bu(x) = A \sum_{k = 1}^N \lambda_k x_k + B \sum_{k = 1}^N \lambda_k u_k \\
\ \ = \sum_{k = 1}^N \lambda_k \left( A x_k + B u_k \right) \in \co \left( 
\bigcup_{k = 1}^{N-1} \Omega_k \cup \Omega\right) \subseteq 
\Omega_{\infty}
\end{array}
\end{equation*}
from condition (\ref{eq:invariance_cond1}). 
\end{proof}

This means that the set given by (\ref{eq:Omegainfty}) is control invariant, in 
the absence of state constraints, if (\ref{eq:invariance_cond1}) is satisfied. 

To take into account the constraints on the state $x \in X$, recall that, under 
Assumption \ref{ass:XU}, if $\Omega$ is a control invariant set, then also 
$\alpha \Omega$ is a control invariant set, in absence of state constraints. 
Thus a first method would consists, given a control invariant set 
$\Omega_{\infty}$ in absence of state constraints, in computing the greatest 
$\alpha \in [0, \, 1]$ such that $\alpha \Omega_{\infty} \subseteq X$. This 
method, together with a less conservative one which takes explicitly into 
account $X$ in the computation of $\Omega_{\infty}$, are illustrated in 
Section~\ref{sec:state_constr}. Both methods are based on the results valid in 
absence of state constraints.

\begin{remark}
The algorithm sketched above is not the standard one for obtaining a control 
invariant set. Usually, in fact, one should start with $\Omega = X$ and 
intersect the preimages with $X$ at every iteration and then check if the 
inclusion $\Omega_k \subseteq \Omega_{k+1}$ holds, see \cite{BlanchiniBook}. 
This approach provides a sequence of non-increasing nested sets that are outer 
approximations of the maximal control invariant set and whose intersection 
converges to it, if $X$ and $U$ are compact, see \cite{BertsekasTAC72}. 
Unfortunately, nevertheless, the maximal control invariant set is in general 
not finitely determined and the sets generated by the iteration are not control 
invariant. An alternative, related to the approach presented here, is to start 
with $\Omega$ that is already control invariant, which leads to a 
non-decreasing sequence of nested control invariant sets. The algorithm 
presented here has the benefit of not requiring the \emph{a priori} knowledge 
of a control invariant set $\Omega$, but, on the other hand, does not assure 
that the stop condition is satisfied at some iteration for a given $\Omega$. A 
scaling procedure will be employed in order to guarantee that the stop condition 
holds.
\end{remark}

Given the initial set $\Omega$, an alternative condition characterizing an 
invariant set is the following 
\begin{equation}\label{eq:invariance_condN}
\Omega \subseteq A^{-N} \left( \Omega \oplus \bigoplus_{i = 0}^{N-1} (-A^i B U 
) \right),
\end{equation}
which is equivalent to the fact that every state in $\Omega_N$ can be steered 
in $\Omega$ in exactly $N$ steps. 

This means that (\ref{eq:invariance_condN}) implies, but is not equivalent to, 
(\ref{eq:invariance_cond1}) and the resulting invariant set would be 
$\Omega_{\infty}$ as in (\ref{eq:Omegainfty}). Condition 
(\ref{eq:invariance_condN}), which will be referred to as N-step condition in 
what follows, is just sufficient for (\ref{eq:invariance_cond1}) to hold but it 
does not require the computation of the convex hull of several sets at every 
iteration. The related algorithm follows, in which the N-step condition and the 
explicit representation of $\Omega_k$ (\ref{eq:Ok}) have been used.

\begin{algorithm}[H]
\caption{N-step condition control invariant}\label{alg:N}
\begin{algorithmic}[1]
\Input matrices $A, B$, sets $\Omega$, $U$
\State $k \gets 0$
\Repeat 
\State $k \gets k+1$
\Until {$\displaystyle \Omega \subseteq A^{-k} \left( \Omega \oplus 
\bigoplus_{i = 0}^{k-1} (-A^i B U ) \right)$}
\State $N \gets k$
\Output $\displaystyle \Omega_{\infty} \gets \co \left(\bigcup_{k = 1}^{N} 
\Omega_k \right)$
\end{algorithmic}
\end{algorithm}

The main issue which impedes the application of both algorithms in high 
dimension is the fact that computing the Minkowski set addition is a complex 
operation, as it is an NP-complete problem, see \cite{Tiwary08Hardness}. 
Moreover the addition leads to sets whose representation complexity increases. 
Considering, in fact, two polytopic sets $\Omega$ and $\Delta$, their sum has 
in general more facets and vertices those of $\Omega$ and $\Delta$. Thus, the 
algorithm given above requires the computation of the Minkowski sum, hardly 
manageable in high dimension, and generates polytope with an increasing number 
of facets and vertices. Another source of complexity is the convex hull in 
(\ref{eq:invariance_cond1}) or (\ref{eq:Omegainfty}), as the explicit 
computation of the convex hull is a non-convex operation whose complexity grows 
exponentially with the dimension, see \cite{DeBerg2000computational}.

The main objective of this paper is to design a method for testing conditions 
(\ref{eq:invariance_cond1}) and (\ref{eq:invariance_condN}) by means of convex 
optimization problems, then applicable also to relatively high dimensional 
systems, for obtaining a control invariant set.

\section{N-step condition for control invariance}

As noticed above, a first main issue is related to check whether the sum of 
several polytopes contains a polytope, see the $N$-step stop condition 
(\ref{eq:invariance_condN}). Then, also the fact that the convex hull 
computation could be required, as in condition (\ref{eq:invariance_cond1}), 
would introduce additional complexity. We consider first the $N$-step stop 
condition used in Algorithm~\ref{alg:N} and the computation of the induced 
control invariant $\Omega_{\infty}$. The stop condition 
(\ref{eq:invariance_cond1}) of Algorithm~\ref{alg:I} is based on these results 
and will be illustrated afterward.

\subsection{Minkowski sum and inclusion}
Consider the N-step condition (\ref{eq:invariance_condN}), characterized by the 
Minkowski sum of several sets. The explicit definition of the Minkowski sum of 
sets could be avoided by employing its implicit representation. Indeed, given 
two polyhedral sets $\Gamma = \{x \in \R^m : H x \in h\}$ and $\Delta = \{y \in 
\R^p: G y \leq g\}$ and $P \in \R^{n \times m}$ and $Q \in \R^{n \times p}$ we 
have that $P \Gamma \oplus T \Delta = \{x \in \R^n: x = P y + T z, \ Hy \leq h, 
\ Gz\leq g \}$. Thus, the explicit hyperplane or vertex representation of the 
sum can be replaced by the implicit one, given by the projection of a polyhedron 
in higher dimension. On the other hand, one might wonder if the stop condition 
$\Omega \subseteq \Omega_{N}$ could be checked without the explicit 
representation of $\Omega_{N}$. 

The first remark to do is that the inclusion condition is testable through a 
set of LP problems provided the vertices of $\Omega$ are available. Such an 
assumption is not very restrictive, since $\Omega$ is a design parameter that 
could be determined such that both the hyperplane and vertices representation 
should be available, a box for instance. Nevertheless, and since we are aiming 
at invariant sets for high dimensional systems, the use of vertices should be 
avoided if possible. Consider for instance, in fact, a system with $n = 20$. The 
unit box in $\R^{20}$ is characterized by 40 hyperplanes, but it has $2^{20} 
\simeq 10^6$ vertices. Then checking if it is contained in a set could require 
to solve more than a million of LP problems.

We consider then the possibility of testing whether a polyhedron is included in 
the sum of polyhedra by employing only their hyperplane representations and 
without the explicit representation of the sum of sets. The following result, 
based on the Farkas lemma and widely used on set theory and invariant methods 
for control, is useful for this purpose.

\begin{lemma}\label{lem:Farkas}
Two polyhedral sets $\Gamma = \{x \in \R^n : H x \leq h\}$, with $F \in 
\R^{p \times n}$, and $\Delta = \{x \in \R^n: G x \leq g\}$, with $G \in 
\R^{q \times n}$, satisfy $\Gamma \subseteq \Delta$ if and only if there 
exists a non-negative matrix $T \in \R^{q \times p}$ such that
\begin{equation*}
\begin{array}{l}
T H = G,\\
T h \leq g.
\end{array}
\end{equation*}
\end{lemma}
Consider now the stop condition (\ref{eq:invariance_condN}), which is suitable 
for applying the Lemma \ref{lem:Farkas}, as illustrated below. 

\begin{remark}\label{eq:setalgebra2}
The right-hand side term in (\ref{eq:invariance_condN}) cannot be expressed 
directly as the Minkowski sum of several sets, unless $A$ is nonsingular. In 
fact, given $A \in \R^{n \times n}$ with $\det(A) \neq 0$ and $\Gamma, 
\Delta \subseteq \R^n$ then the matrix $A^{-1}$ is defined and thus
\begin{equation}\label{eq:A-1OmegaB}
A^{-1} \Omega = \{x \in \R^n: \ A x \in \Omega\} = \{A^{-1} x \in \R^n : x \in 
\Omega\}, 
\end{equation}
which implies that 
\begin{equation*}
 \begin{array}{l}
A^{-1}(\Gamma \oplus \Delta) = \{x \in \R^n : A x \in \Gamma \oplus 
\Delta\} = \{x \in \R^n : A x = \\
\hspace{0.2cm} y + z, \ y \in \Gamma, \ z \in \Delta\} = \{x \in \R^n : {x} = 
A^{-1} y + A^{-1} z, \\
\hspace{0.2cm} y \in \Gamma, \ \ z \in \Delta\} = \{A^{-1} y + A^{-1} z \in \R^n 
: \ y \in \Gamma, \ \ z \in \Delta\} \\
\hspace{0.2cm} = \{y + z \in \R^n : \ A y \in \Gamma, \ \ A z \in \Delta\} = 
A^{-1}\Gamma \oplus A^{-1} \Delta,
 \end{array}
\end{equation*}
since the matrix $A^{-1}$ exists. On the contrary, if $\det(A) = 0$ then we 
have that $A^{-1}(\Gamma \oplus \Delta) \neq A^{-1}\Gamma \oplus A^{-1} 
\Delta$ in general. Indeed, considering for instance 
\begin{equation*}
\begin{array}{l}
 \Gamma = \{x \in \R^2: \ 1 \leq x_{(1)} \leq 2, \quad -1 \leq x_{(2)} \leq 
1\},\\
 \Delta = \{x \in \R^2: \ -3 \leq x_{(1)} \leq -1, \ \ -1 \leq x_{(2)} \leq 
1\},\\
\end{array}
\end{equation*}
and $A = \left[\begin{array}{cc} 0 & 0\\ 0 & 1 \end{array}\right]$, it follows 
that $ A^{-1} \Gamma = A^{-1} \Delta = \emptyset $ but 
\begin{equation*}
\begin{array}{l}
 \Gamma \oplus \Delta = \{x \in \R^2: \ -2 \leq x_{(1)} \leq 1, \quad -2 
\leq x_{(2)} \leq 2\},\\
A^{-1} (\Gamma \oplus \Delta ) = \{x \in \R^2: \ -2 \leq x_{(2)} \leq 2\}.
\end{array}
\end{equation*}
\end{remark}

The main issue for applying Lemma \ref{lem:Farkas} is the fact that obtaining 
the explicit hyperplane representation of the set at right-hand side of 
(\ref{eq:invariance_condN}) is numerically hardly affordable, mainly in 
relatively high dimension. In fact, given two polyhedra $\Gamma \subseteq \R^m$ 
and $\Delta \subseteq \R^p$, to determine $L$ and $l$ such that $P \Gamma \oplus 
Q \Delta = \{x \in \R^n : \ L x \leq l\}$ is an NP-complete problem, see 
\cite{Tiwary08Hardness}. Nevertheless, a sufficient condition in form of LP 
feasibility problem is given below for testing if a polyhedral set $\Omega$ is 
contained in $P \Gamma \oplus Q \Delta$.

\begin{proposition}\label{prop:Mink}
Consider the sets $\Omega = \{x \in \R^n : H x \leq h\}$, $\Gamma = \{y \in \R^m 
: F y \leq f\}$, $\Delta = \{z \in \R^p : G z \leq g\}$ and with $H \in \R^{n_h 
\times n}$, $F \in \R^{n_f \times m}, G \in \R^{n_g \times m}$ and the matrices 
$P \in \R^{n \times m}$ and $Q \in \R^{n \times p}$. Then $\Omega \subseteq P 
\Gamma \oplus Q \Delta$ if there exist $T \in \R^{n_{\bar{g}} \times 
n_{\bar{h}}}$ and $M \in \R^{(n+m+p) \times (n+m+p)}$, with $n_{\bg} = 2n + n_f 
+ n_G$ and $n_{\bh} = n_h + 2m + 2p$ such that 
\begin{equation}\label{eq:FinalCond}
\left\{\begin{array}{l}
T \bH = \bG M\\
T \bh \leq \bg\\
\left[\begin{array}{ccc} I \ee 0 \ee 0 \end{array}\right] =  
\left[\begin{array}{ccc} I \ee 0 \ee 0 \end{array}\right] M
\end{array}\right.
\end{equation}
holds with 
\begin{equation}\label{eq:bG}
\bG = \left[\begin{array}{ccc}
I & -P & -Q\\
-I & P & Q\\
0 & F & 0\\
0 & 0 & G
\end{array}\right] \in \R^{n_{\bg} \times (n+m+p)} \quad 
\bg = \left[\begin{array}{c}
0\\
0\\
f\\
g
\end{array}\right]\in \R^{n_{\bg}}.
\end{equation}
and 
\begin{equation}\label{eq:bH}
\bH = \left[\begin{array}{ccc}
H & 0 & 0\\
0 & I & 0\\
0 & -I & 0\\
0 & 0 & I\\
0 & 0 & -I
\end{array}\right] \in \R^{n_{\bh} \times (n+m+p)} \quad \bh = 
\left[\begin{array}{c}
h\\
0\\
0\\
0\\
0
\end{array}\right] \in \R^{n_{\bh}}.
\end{equation}
\end{proposition}

\begin{proof}
The Minkowski sum of $P \Gamma$ and $Q \Delta$ has an implicit hyperplane 
representation given by 
\begin{equation*}
P \Gamma \oplus Q \Delta = \{x \in \R^n : \ x = P y + Q z, \ y \in \Gamma, \ z 
\in \Delta \} \subseteq \R^n
\end{equation*}
which is equivalent to the projection on $\R^n$ of a polyhedron in 
$\R^{n+m+p}$, that is 
\begin{equation}\label{eq:PO1QO2}
P \Gamma \oplus Q \Delta = \proj_x \ \Omega_{\oplus}
\end{equation}
where $\proj_x$ is the projection on the subspace of $x$, i.e. $\proj_x \ 
\Omega_{\oplus} = [I \ 0 \ 0] \Omega_{\oplus}$, and 
\begin{equation*}\label{eq:OmegaplusZ}
\begin{array}{rl}
\Omega_{\oplus} \hspace{-0.2cm} & = \{(x, y, z) \in \R^{n+m+p}: \ x = P y + Q 
z, \ F y \leq f, \ G z \leq g\}\\
& = \{\bx \in \R^{n+m+p}: \ \bG \bx \leq \bg\} \subseteq \R^{n+m+p},
\end{array}
\end{equation*}
with $\bx = (x, y, z) \in \R^{n+m+p}$ and $\bG, \bg$ as in (\ref{eq:bG}). Thus, 
to prove that $\Omega \subseteq P \Gamma \oplus Q \Delta$ without computing the 
hyperplane representation of the set $P \Gamma \oplus Q \Delta$ is equivalent to 
check whether the projection of $\Omega_{\oplus}$ on $\R^n$ contains $\Omega 
\subseteq \R^n$. This is equivalent to consider the set 
\begin{equation*}\label{eq:Omegainfty0Z}
\begin{array}{rl}
\bar{\Omega} \hspace{-0.2cm} & = \Omega \times \{0\} \times \{0\} = \{(x, 
y, z) \in \R^{n+m+p}: \ H x \leq h, \\
& y = 0, \ z = 0\} = \{\bx \in \R^{n+m+p}: \ \bH \bx \leq \bh\} \subseteq 
\R^{n+m+p}
\end{array}
\end{equation*}
with $\bx = (x, y, z)$ and $\bH, \bh$ as in (\ref{eq:bH}), and test if 
\begin{equation}\label{eq:projCNS}
\proj_x \bar{\Omega} \subseteq \proj_x \Omega_{\oplus},
\end{equation}
since $\Omega = \proj_x \bar{\Omega}$ and from (\ref{eq:PO1QO2}). Unfortunately, 
condition (\ref{eq:projCNS}) is not suitable for using Lemma \ref{lem:Farkas} 
and then we search for a sufficient condition for (\ref{eq:projCNS}) to hold 
such that the lemma can be applied directly. 

Consider any linear single-valued mapping $M : \R^{n+m+p} \rightrightarrows 
\R^{n+m+p}$, characterized by a, possibly non-invertible, matrix $M \in 
\R^{(n+m+p) \times (n+m+p)}$, such that the value of $x$ through $M$ is 
preserved, i.e. $\proj_x M((x, y, z)) = x$ for all $(x, y, z) \in \R^{n+m+p}$. 
Clearly, the value of $x$ is preserved also through the inverse mapping of $M$, 
that is $\proj_x M^{-1}((x, y, z)) = x$ for all $(x, y, z) \in \R^{n+m+p}$. 
This means that $\proj_x \Omega_{\oplus} = \proj_x M^{-1} \Omega_{\oplus}$ and 
then (\ref{eq:projCNS}) is equivalent to 
\begin{equation}\label{eq:projCNS2}
\proj_x \bar{\Omega} \subseteq \proj_x M^{-1} \Omega_{\oplus}.
\end{equation}
Then, the existence of $M$ preserving the $x$ and such that 
\begin{equation}\label{eq:projCNS3}
\bar{\Omega} \subseteq M^{-1} \Omega_{\oplus}
\end{equation}
holds, is a sufficient condition for (\ref{eq:projCNS2}), and thus also for 
(\ref{eq:projCNS}), to be satisfied. Notice that necessity of 
(\ref{eq:projCNS3}) for (\ref{eq:projCNS2}) is not straightforward, since 
$\proj_x \Gamma \subseteq \proj_x \Delta$ does not imply $\Gamma 
\subseteq \Delta$, in general. 

The condition on the matrix $M$ such that $\proj_x M((x, y, z)) = x$ for all 
$(x, y, z) \in \R^{n+m+p}$ is 
\begin{equation}\label{eq:N}
\left[\begin{array}{ccc} I \ee 0 \ee 0 \end{array}\right] =  
\left[\begin{array}{ccc} I \ee 0 \ee 0 \end{array}\right] M
\end{equation}
and then, from Lemma \ref{lem:Farkas} and Remark (\ref{rem:setalgebra1}), it 
follows that 
conditions (\ref{eq:projCNS3}) and (\ref{eq:N}) are equivalent to the existence 
of $T \in \R^{n_{\bar{g}} \times n_{\bar{h}}}$ and $M \in \R^{(n+m+p) \times 
(n+m+p)}$ satisfying (\ref{eq:FinalCond}). Then (\ref{eq:FinalCond}) is a 
sufficient condition for $\Omega \subseteq P \Gamma \oplus Q \Delta$.
\end{proof}

Thus, the inclusion of a set in the sum of sets can be tested by solving an LP 
feasibility problem. This results is applied to the stop condition for control 
invariance.

\subsection{N-step invariance condition as an LP problem}

Consider now condition (\ref{eq:invariance_condN}) with 
\begin{equation}\label{eq:OmegaU}
\Omega = \{x \in \R^n : \ H x\leq h\}, \quad U = \{u \in \R^m : \ G u \leq g\}
\end{equation}
where $H \in \R^{n_h \times n}$ and $G \in \R^{n_g \times m}$. 
Following the reasonings of the proof of Proposition~\ref{prop:Mink}, a 
tractable condition for the set inclusion (\ref{eq:invariance_condN}) to hold 
is given.

\begin{theorem}\label{th:Ninclusion}
Consider $\Omega$ and $U$ as in (\ref{eq:OmegaU}), with $H \in \R^{n_h \times 
n}$ and $G \in \R^{n_g \times m}$, and suppose that $0 \in \Omega$ and $0 \in 
U$. Then the set $\Omega_\infty$ as in (\ref{eq:Omegainfty}) is a control 
invariant set if there exist $T \in \R^{n_{\bar{g}} \times n_{\bar{h}}}$ and $M 
\in \R^{\bar{n} \times \bar{n}}$, with $n_{\bg} = n_h+Nn_g$, $n_{\bh} = n_h+2Nm$ 
and $\bar{n} = n+Nm$, such that 
\begin{equation}\label{eq:FinalCondLP}
\left\{\begin{array}{l}
T \bH = \bG M\\
T \bh \leq \bg\\
\left[\begin{array}{ccccc} I \ee 0 \ee 0 \ee \ldots \ee 0 \end{array}\right] =  
\left[\begin{array}{ccccc} I \ee 0 \ee 0 \ee \ldots \ee 0\end{array}\right] M
\end{array}\right.
\end{equation}
hold with 
\begin{equation}\label{eq:bGinv}
\bG = \left[\begin{array}{ccccc}
H A^N \ee H B \ee HAB \ee \ldots \ee HA^{N-1}B\\
0 \ee G \ee 0 \ee \ldots \ee 0\\
0 \ee 0 \ee G \ee \ldots \ee 0\\
 \ldots \ee \ldots  \ee \ldots  \ee \ldots \ee 
\ldots \\
0 \ee 0 \ee 0 \ee \ldots \ee G\\
\end{array}\right] \hspace{-0.1cm}, \ \
\bg = 
\left[\begin{array}{c}
h\\
g\\
g\\
\ldots\\
g
\end{array}\right] 
\end{equation}
where $\bG \in \R^{n_{\bg} \times \bar{n}}$ and $\bg \in \R^{n_{\bg}}$, and 
\begin{equation}\label{eq:bHinv}
\bH = \left[\begin{array}{cccccc}
H \ee 0 \ee 0 \ee \ldots \ee 0\\
0 \ee I \ee 0 \ee \ldots \ee 0\\
0 \ee -I \ee 0 \ee \ldots \ee 0\\
0 \ee 0 \ee I \ee \ldots \ee 0\\
0 \ee 0 \ee -I \ee \ldots \ee 0\\
 \ldots \ee \ldots  \ee \ldots  \ee \ldots \ee \ldots \\
0 \ee 0 \ee 0 \ee \ldots \ee I\\
0 \ee 0 \ee 0 \ee \ldots \ee -I\\
\end{array}\right] \hspace{-0.1cm}, \quad
\bh = \left[\begin{array}{c}
h\\
0\\
0\\
0\\
0\\
\ldots\\
0\\
0
\end{array}\right]
\end{equation}
where $\bH \in \R^{n_{\bh} \times \bar{n}}$ and $\bh \in \R^{n_{\bh}}$.
\end{theorem}

\begin{proof}
Consider condition (\ref{eq:invariance_condN}), sufficient for $\Omega_{\infty}$ 
to be a control invariant set. The proof follows the lines of the one of 
Proposition~\ref{prop:Mink}. From Remark~\ref{rem:setalgebra1}, the right-hand 
side term of (\ref{eq:invariance_condN}) is given by 
\begin{equation*}
\begin{array}{l}
\displaystyle \Omega_N = A^{-N} \hspace{-0.1cm} \left( \hspace{-0.1cm} \Omega 
\oplus \bigoplus_{i = 0}^{N-1} (-A^i B U) \hspace{-0.1cm} \right) = \{x \in 
\R^n : \, A^N x = y - B u_1\\
\hspace{0.4cm} - AB u_2 - \ldots -A^{N-1} Bu_N, \ H y \leq h, \ G u_i \leq g  
\ \forall i \in \N_N \}\\
\hspace{0.4cm} = \{x \in \R^n : \, H A^N x + H B u_1 + H A^{1} Bu_2 + \ldots \\
\hspace{0.4cm} + H A^{N-1} Bu_N  \leq h, \ G u_i \leq g \ \forall i \in 
\N_N \}
\end{array}
\end{equation*}
and then is the projection on $\R^n$ of the set 
\begin{equation*}\label{eq:OplusinvZ}
\begin{array}{l}
\Omega_{\oplus} = \{(x, u_{1}, u_{2}, \ldots, u_{N}) \in \R^{\bar{n}}: \ H A^N x 
+ H B u_1 + H A B  u_2 \ldots \\
+ H A^{N-1}Bu_N \leq h, \ G u_i \leq g \ \forall i 
\in \N_N\} = \{\bx \in \R^{\bar{n}}: \ \bG \bx \leq \bg\},
\end{array}
\end{equation*}
with $\bx = (x, u_{1}, u_{2}, \ldots, u_{N}) \in \R^{\bar{n}}$ and $\bG$ and 
$\bg$ as in (\ref{eq:bGinv}). The set $\bO$ in this case would result in 
\begin{equation*}
\begin{array}{cl}
\bO \hspace{-0.2cm} & = \{(x, u_{1}, u_{2}, \ldots, u_{N}) \in \R^{\bar{n}}: \ H 
x \leq h, \ u_i = 0 \ \forall i \in \N_N\} \\
& = \{\bx \in \R^{\bar{n}}: \ \bH \bx \leq \bh\} \subseteq \R^{\bar{n}}
\end{array}
\end{equation*}
with $\bH$ and $\bh$ as in (\ref{eq:bHinv}). From Proposition~\ref{prop:Mink}, 
the condition (\ref{eq:invariance_condN}) is satisfied if there are $T \in 
\R^{n_{\bar{g}} \times n_{\bar{h}}}$ and $M \in \R^{\bar{n} \times \bar{n}}$ 
such that (\ref{eq:FinalCondLP}) holds
\end{proof}

Finally, given the set $\Omega$ and $U$, to obtain the greatest multiple of 
$\Omega$, i.e. $\Omega_{\alpha} = \alpha \Omega$ such that 
(\ref{eq:invariance_condN}) holds, that is the greatest 
$\alpha \in \R$ such that 
\begin{equation}\label{eq:invariance_cond2}
\alpha \Omega = \Omega^{\alpha} \subseteq \Omega_N^{\alpha},
\end{equation}
with 
\begin{equation}\label{eq:Omegaka}
\Omega_k^{\alpha} = A^{-k} \left( \Omega^{\alpha} \oplus \bigoplus_{i = 
0}^{k-1} 
(-A^i B U)\right), \quad \forall k \in \N,\end{equation}
is equivalent to compute the smallest nonnegative $\beta$, with $\beta = 
\alpha^{-1}$, such that 
\begin{equation*}\label{eq:invariance_cond3Z}
 \Omega \subseteq A^{-N} \left( \Omega \oplus \bigoplus_{i = 0}^{N-1} (-A^i B 
\beta U) \right).
\end{equation*}
This consists in replacing $g$ with $\beta g$ in (\ref{eq:bGinv}) and leads to 
the following LP problem in $T$, $M$ and $\beta$ 
\begin{equation}\label{eq:FinalCondLPbeta}
\begin{array}{l}
\hspace{-0.5cm} \alpha^{-1} = \beta_N = \min_{\beta \in \R_+} \beta\\
\hspace{1.3cm} \mathrm{s.t. } \ \ T \bH = \bG M\\
\hspace{2.0cm} T \bh \leq \beta \hat{g} + \tilde{g}\\
\hspace{2.0cm} \left[\begin{array}{ccccc} I \ee  0 
\ee 0 \ee  \ldots 
\ee 0 \end{array}\right] =  
\left[\begin{array}{ccccc} I \ee  0 
\ee 0 \ee  \ldots 
\ee 0\end{array}\right] M
\end{array}
\end{equation}
with $\hat{g} = (0, \ g, \ g, \  \ldots, \ g)$ and $\tilde{g} = (h, \ 0, \ 0, \ 
\ldots, \ 0)$, sufficient for the N-step invariant 
condition 
\begin{equation*}\label{eq:invariance_condSZ}
\beta_N^{-1} \Omega \subseteq A^{-N} \left( \beta_N^{-1} \Omega \oplus 
\bigoplus_{i = 0}^{N-1} (-A^i B U) \right) \quad \mathrm{ with } \quad \alpha = 
\beta_N^{-1},
\end{equation*}
to hold. Note that using directly $\alpha$ would yield to replacing $h$ by 
$\alpha h$ in (\ref{eq:bGinv}) and (\ref{eq:bHinv}) and then to a nonlinear 
optimization problem. 

\section{Control invariant set\\ and state constraints}

If the stop condition (\ref{eq:invariance_condN}) is satisfied after 
appropriately scaling $\Omega$, i.e. with $\Omega = \Omega^{\alpha}$ satisfying 
(\ref{eq:FinalCondLPbeta}), the set $\Omega_N^{\alpha}$ is such that if $x \in 
\Omega_N^{\alpha}$ then it can be steered in $\Omega^{\alpha}$ in $N$ steps by a 
sequence of admissible control inputs $u_i \in U$ with $i \in \N_{N}$. Recall 
that, until now, the constraints on the state have not been taken into account, 
they will in Section~\ref{sec:state_constr}. 

Once $\Omega^{\alpha}$ is computed, one possible choice to obtain a control 
invariant set is considering $\Omega_{\infty}$ as in (\ref{eq:Ok}) and 
(\ref{eq:Omegainfty}). This would require to compute the convex hull of 
the union of several sets, each one given by the Minkowski sum of sets, but the 
convex hull operation is numerically demanding.

For this, given an arbitrary collection of non-empty convex sets 
$\Gamma_i \subseteq \R^n$ with $I \in \N$ and $i \in \N_I$, note that 
\begin{equation*}
\begin{array}{l}
\co\left(\bigcup_{i \in \N_I} \Gamma_i \right) = \bigcup_{\substack{\lambda \geq 
0 \\ \mathbf{1}^T \lambda = 1}} \hspace{-0.2cm} \left( \bigoplus_{i \in \N_I} 
\lambda_i \Gamma_i \right) \quad \mathrm{ and } \quad \displaystyle \lambda 
\bigoplus_{i \in \N_I}  \Gamma_i = \bigoplus_{i \in \N_I} \lambda \Gamma_i,
\end{array}
\end{equation*}
see Chapter 3 in \cite{Rockafellar}. Then, provided condition 
(\ref{eq:invariance_cond2}) is satisfied and with definition of $\Omega$ and 
$U$ as in (\ref{eq:OmegaU}), the invariant set is given by 
\begin{equation}\label{eq:Omegaconvex}
\begin{array}{rl}
 \Omega_\infty^{\alpha} \hspace{-0.2cm} & \displaystyle = \co \left(\bigcup_{k 
= 1}^{N} \Omega_k^{\alpha} \right) = \bigcup_{\substack{\lambda \geq 0 \\ 
\mathbf{1}^T \lambda = 1}} \left( \bigoplus_{k = 1}^N \lambda_k 
\Omega_k^{\alpha} \right) \\
& \displaystyle = \bigcup_{\substack{\lambda \geq 0 \\ 
\mathbf{1}^T \lambda = 1}} \left( \bigoplus_{k = 1}^N \lambda_k A^{-k} \left( 
\Omega^{\alpha} \oplus \bigoplus_{i = 0}^{k-1} (-A^i B U) \right)\right).
\end{array}
\end{equation}
Before proceeding, it is essential to notice that, given a convex set $\Omega$, 
the set $\gamma A^{-1} \Omega$ is well defined for all $\gamma \in \R$ and $A 
\in \R^{n \times n}$, even for $\gamma = 0$ and singular matrices $A$. In fact, 
it is given by
\begin{equation*}
 \gamma A^{-1} \Omega = \{\gamma x \in \R^n: \ x \in A^{-1} \Omega \} = 
\{\gamma 
x \in \R^n: \ A x \in \Omega \}.
\end{equation*}
This means, for instance, that, for $n = 1$, if $\gamma = 0$ and $A = 0$ one 
has $\gamma A^{-1} \Omega = \{0\}$ if $0 \in \Omega$, even if $A^{-1} \Omega = 
\R$. 

\begin{lemma}\label{lem:gammaA}
For every $\Omega \subseteq \R^n$, if $\gamma \neq 0$, with $\gamma \in \R$, or 
$\det(A) \neq 0$ then $\gamma A^{-1} \Omega = A^{-1} \gamma \Omega$.
\end{lemma}

\begin{proof}
If $\gamma \neq 0$, one has 
\begin{equation*}\label{eq:switchgammaZ}
\begin{array}{l}
x \in \gamma A^{-1} \Omega \ \ \Leftrightarrow \ \  x = \gamma y, \ \ y \in 
A^{-1} \Omega \ \ \Leftrightarrow \ \  x = \gamma y, \ \ A y \in \Omega \\
\Leftrightarrow \ \ \gamma^{-1} x = y, \ \ A y \in \Omega \ \ \Leftrightarrow 
\ \ A \gamma^{-1} x \in \Omega \ \ \Leftrightarrow \ \ \gamma^{-1} A x \in 
\Omega \\ 
\Leftrightarrow \ \ A x = \gamma y, \ \ y \in \Omega \ \ \ \ \Leftrightarrow \ 
\ 
A x \in \gamma \Omega \ \ \Leftrightarrow \ \ x \in A^{-1} \gamma \Omega,
\end{array}
\end{equation*}
whereas if $\det(A) \neq 0$ it follows that 
\begin{equation*}
\begin{array}{l}
x \in \gamma A^{-1} \Omega \Leftrightarrow \ x = \gamma y, \ y \in A^{-1} 
\Omega \ \Leftrightarrow \ x = \gamma y, \ A y \in \Omega \ 
\Leftrightarrow \\
x = \gamma A^{-1} z, \ z \in \Omega \ \Leftrightarrow \ x = A^{-1} \gamma z, 
\ z \in \Omega \ \Leftrightarrow \ x \in A^{-1} \gamma \Omega.
\end{array}
\end{equation*}
\end{proof}

This means, in practice, that the operators $\gamma$ and $A^{-1}$ actuating 
on $\Omega$ can be switched, if either $\gamma \neq 0$ or $\det(A) \neq 0$. Note 
that, if $\gamma = 0$ and $A$ is singular, then the equality $\gamma A^{-1} 
\Omega = A^{-1} \gamma \Omega$ does not hold in general, as illustrated in the 
following example.

\begin{example}
Consider $\lambda = 0$,  $A = \left[\begin{array}{cc} 0 & 0\\ 0 & 1 
\end{array}\right]$ and $\Omega = \B^n$. Then $A^{-1} \gamma \Omega = \{x \in 
\R^n : A x \in \{0\}\} = \{x \in \R^n :  x_{(2)} = 0\}$ and $\gamma A^{-1} 
\Omega = \{\lambda x \in \R^n :  -1 \leq x_{(2)} \leq 1\} = \{0\}$.
\end{example} 

The cases of nonsingular and singular matrix $A$ are considered individually. 

\subsection{Invariant for nonsingular $A$}

If $A$ is nonsingular the invariant set is the polyhedron give below. 

\begin{proposition}
Let $\Omega$ and $U$ as in (\ref{eq:OmegaU}). If $\det(A) \neq 0$ then 
$\Omega_{\infty}^{\alpha}$ defined in (\ref{eq:Omegaconvex}) is equal to 
$\check{\Omega}_{\infty}^{\alpha}$ where
\begin{equation}\label{eq:Inv_for_Anonsin}
\begin{array}{l}
 \check{\Omega}_\infty^{\alpha} = \{ x \in \R^n: x = \sum_{k = 1}^N z_k, \, H 
A^{k} z_k + \sum_{i = 0}^{k-1} H A^i B v_{i+1,k} \leq \lambda_k \alpha_N h, \\
\hspace{1cm} G v_{i,k} \leq \lambda_k g \ \forall i \in \N_k \ \forall k \in 
\N_N, \ \lambda \geq 0, \ \sum_{k = 1}^N \lambda_k = 1\}.
\end{array}
\end{equation}
\end{proposition}

\begin{proof}
From (\ref{eq:Omegaconvex}) and Lemma \ref{lem:gammaA} it follows
\begin{equation*}
\begin{array}{l}
\Omega_\infty^{\alpha} = \{ x \in \R^n: \ x = \sum_{k = 1}^N 
z_k, \ z_k \in \lambda_k \Omega_k^{\alpha} \ \forall k \in \N_N, \ \lambda \geq 
0, \\
\displaystyle \ \ \sum_{k = 1}^N \lambda_k = 1\} = \{ x \in \R^n: \ x = \sum_{k 
= 1}^N z_k, \ z_k \in A^{-k} \lambda_k \Bigg(\Omega^{\alpha}\\
\ \oplus \bigoplus_{i = 0}^{k-1} (-A^i B U) \Bigg) \forall k \in \N_N, \ 
\lambda 
\geq 0, \ \sum_{k = 1}^N \lambda_k = 1\} = \{ x \in \R^n: \\
\ \ x = \sum_{k = 1}^N z_k, \ \ A^{k} z_k \in  \left(\lambda_k  \Omega^{\alpha} 
\oplus \bigoplus_{i = 0}^{k-1} (-A^i B \lambda_k U) \right) \forall k \in 
\N_N,\\
\ \ \lambda \geq 0, \ \sum_{k = 1}^N \lambda_k = 1\} = \{ x \in \R^n: \ x = 
\sum_{k = 1}^N z_k, \ A^{k} z_k = y_k \\
\ \ - \sum_{i = 0}^{k-1} A^i B v_{i+1,k}, \ y_k \in  \lambda_k \Omega^{\alpha}, 
\ v_{i,k} \in \lambda_k U \ \forall i \in \N_k \ \forall k \in \N_N, \\
\ \ \lambda \geq 0, \ \sum_{k = 1}^N \lambda_k = 1\} = \{ x \in \R^n: \ x = 
\sum_{k = 1}^N z_k, \ A^{k} z_k = y_k - \\
\ \ \sum_{i = 0}^{k-1} A^i B v_{i+1,k}, \ H y_k \leq \lambda_k \alpha_N h, \ G 
v_{i,k} \leq \lambda_k g \ \forall i \in \N_k \\
\ \ \forall k \in \N_N, \ \lambda \geq 0, \ \sum_{k = 1}^N \lambda_k = 1\} = 
\check{\Omega}_\infty^{\alpha} 
\end{array}
\end{equation*}
where the second equality holds since $A$ is nonsingular and then $\lambda_k$ 
and $A^{-k}$ can be switched. 
\end{proof}

Note that the invariant set $\Omega_{\infty}^{\alpha}$ is then characterized by 
linear equalities and inequalities, that is by a polytope in higher dimension. 
This means that testing if a state is in $\Omega_{\infty}^{\alpha}$ reduces to 
solve a feasibility problem with linear constraints. Also the problem of 
enforcing state constraints, see Section~\ref{sec:state_constr} below, can be 
solved through convex optimization by using the representation 
(\ref{eq:Inv_for_Anonsin}). Moreover, such a representation is particularly 
suitable to be used in optimization-based control, as model predictive control 
for instance, since it reduces to enforcing the linear constraints 
characterizing $\Omega_{\infty}^{\alpha}$.

\subsection{Invariant for singular $A$}
In the other case, namely if $A$ is singular, the polyhedral form of the 
invariant set is less straightforward. 

\begin{proposition}\label{prop:Inv_for_Anonsin}
Let $\Omega$ and $U$ as in (\ref{eq:OmegaU}). If $\det(A) = 0$ then 
$\Omega_{\infty}^{\alpha}$ defined in (\ref{eq:Omegaconvex}) is equal to 
$\hat{\Omega}_{\infty}^{\alpha}$ where
\begin{equation}\label{eq:Inv_for_Asin}
\begin{array}{l}
\hat{\Omega}_\infty^{\alpha} = \{ x \in \R^n: \, x = \sum_{k = 1}^N \lambda_k 
w_k, \, H A^{k} w_k + \sum_{i = 0}^{k-1} H A^i B v_{i+1,k} \leq \alpha_N h, \\ 
\hspace{1cm} G v_{i,k} \leq g \ \forall i \in \N_k \, \forall k \in \N_N, \ 
\lambda \geq 0, \ \sum_{k = 1}^N \lambda_k = 1\}.
\end{array}
\end{equation}
\end{proposition}

\begin{proof}
The set $\Omega_{\infty}^{\alpha}$ is given by
\begin{equation*}
\begin{array}{l}
 \Omega_\infty^{\alpha}  \displaystyle = \{ x \in \R^n: \ x = \sum_{k = 1}^N 
z_k, \ z_k \in \lambda_k \Omega_k^{\alpha} \ \forall k \in \N_N, \lambda 
\geq 0,\\
\ \ \sum_{k = 1}^N \lambda_k = 1\} = \{ x \in \R^n: \ x = \sum_{k = 1}^N z_k, \ 
z_k \in \lambda_k A^{-k} \Bigg(\Omega^{\alpha}\\
\ \ \oplus \bigoplus_{i = 0}^{k-1} (-A^i B U) \Bigg) \, \forall k \in \N_N, \ 
\lambda \geq 0, \sum_{k = 1}^N \lambda_k = 1\} = \{ x \in \R^n: \\
\ \ x = \sum_{k = 1}^N z_k, \ z_k = \lambda_k w_k,  \ w_k \in A^{-k} 
\left(\Omega^{\alpha} \oplus \bigoplus_{i = 0}^{k-1} (-A^i B U) \right)\\
\ \  \forall k \in \N_N, \lambda \geq 0, \ \sum_{k = 1}^N \lambda_k = 1\} = \{ 
x 
\in \R^n: \ x = \sum_{k = 1}^N \lambda_k w_k, \\
\ \ A^{k} w_k = y_k - \sum_{i = 0}^{k-1} A^i B v_{i+1,k}, \ y_k \in 
\Omega^{\alpha}, \ v_{i,k} \in U \ \forall i \in \N_k\\
\ \  \forall k \in \N_N, \lambda \geq 0, \sum_{k = 1}^N \lambda_k = 1\} = \{ x 
\in \R^n: \ x = \sum_{k = 1}^N \lambda_k w_k, \\
\ \ A^{k} w_k = y_k - \sum_{i = 0}^{k-1} A^i B v_{i+1,k}, \ H y_k \leq \alpha_N 
h, \ G v_{i,k} \leq g \ \forall i \in \N_k \\
\ \ \forall k \in \N_N, \lambda \geq 0, \sum_{k = 1}^N \lambda_k = 1\} = 
\hat{\Omega}_{\infty}^{\alpha}.
\end{array}
\end{equation*}
\end{proof}

Unfortunately, this representation of $\Omega_{\infty}^{\alpha}$ is not 
suitable to be directly tested through an LP feasibility problem, due to the 
nonlinearities $\lambda_k w_k$. This means that checking whether a state is 
contained in $\hat{\Omega}_\infty^{\alpha}$ could not be solved through LP 
problems, as for nonsingular $A$. Neither the problem of computing the biggest 
copy of $\hat{\Omega}_\infty^{\alpha}$ satisfying the state constraints 
(treated in Section~\ref{sec:state_constr}, see (\ref{eq:LPOinX})) could be 
addressed by convex optimization problems.

What we are going to prove is that the expression of $\Omega_{\infty}^{\alpha}$ 
as in (\ref{eq:Inv_for_Anonsin}) holds also when $A$ is a singular matrix, 
that is $\Omega_{\infty}^{\alpha} = \hat{\Omega}_{\infty}^{\alpha} = 
\check{\Omega}_{\infty}^{\alpha}$. For notational simplicity we define 
\begin{equation}\label{eq:bOk}
\bO_k = \Omega^{\alpha} \oplus \bigoplus_{i = 0}^{k-1} (-A^i B U),
\end{equation}
so that $\Omega_k^{\alpha} = A^{-k} \bO_k$, for all $k \in \N_N$. Then the sets 
$\check{\Omega}_\infty^{\alpha}$ and $\hat{\Omega}_\infty^{\alpha}$ defined in 
(\ref{eq:Inv_for_Anonsin}) and (\ref{eq:Inv_for_Asin}) become
\begin{equation}\label{eq:Oinfs}
\begin{array}{l}
\check{\Omega}_\infty^{\alpha} = \bigcup_{\substack{\lambda \geq 0 \\ 
\mathbf{1}^T \lambda = 1}} \left( \bigoplus_{k = 1}^N A^{-k} \lambda_k \bO_k 
\right), \\
\hat{\Omega}_\infty^{\alpha} = \co \left(\bigcup_{k = 1}^{N} 
A^{-k} \bO_k \right) = \bigcup_{\substack{\lambda \geq 0 \\ 
\mathbf{1}^T \lambda = 1}} \left( \bigoplus_{k = 1}^N \lambda_k A^{-k} \bO_k 
\right),  
\end{array}
\end{equation}
which are equal if $A$ is nonsingular. We prove that they are equal also for 
singular $A$. For this aim, some preliminary results are to be recalled or 
introduced here.

\begin{definition}\label{def:bert}\cite{Bertsekas09}
Given a nonempty convex set $C$, the vector $d$ is a direction of recession of 
$C$ if $x + \alpha d \in C$ for all $x \in C$ and $\alpha \geq 0$. The set of 
all directions of recession is a cone containing the origin, called the 
recession cone of C. The lineality space of C, denoted $L_C$, is the set of 
directions of recession $d$ whose opposite, $-d$, are also directions of 
recession. Given a subspace $S$, $S^{\bot}$ is its orthogonal complement.
\end{definition}

\begin{theorem}\label{th:Rock2_6}\cite{Rockafellar}
A subset of $\R^n$ is a convex cone if and only if it is closed under addition 
and positive scalar multiplication.
\end{theorem}

\begin{theorem}\label{th:Rock3_8}\cite{Rockafellar}
If $K_1$ and $K_2$ are convex cones containing the origin then $K_1 \oplus K_2 
= \co \left(K_1 \cup K_2\right)$.
\end{theorem}

\begin{lemma}\label{lem:S1S2}
Given the subspaces $S_1, S_2 \subseteq \R^n$, we have $S_1 = S_1 \oplus S_1$ 
and $S_1 \oplus S_2 = \co \left(S_1 \cup S_2\right)$. 
\end{lemma}

\begin{proof}
It follows from Theorems~\ref{th:Rock2_6} and \ref{th:Rock3_8} and the fact 
that every subspace is a convex cone containing the origin.
\end{proof}

\begin{proposition}\label{prop:decomposition}(Decomposition of a Convex Set 
\cite{Bertsekas09}) Let $C$ be a nonempty convex subset of $\R^n$. Then, for 
every subspace $S$ that is contained in the lineality space $L_C$, we have $C = 
(C \cap S^{\bot}) \oplus S$.
\end{proposition}

\begin{lemma}\label{lem:CSS}
Given the nonempty convex set $C \subseteq \R^n$, for every subspace $S 
\subseteq C$, we have $C \oplus S = C$. 
\end{lemma}

\begin{proof}
From Proposition~\ref{prop:decomposition} and Lemma~\ref{lem:S1S2}, and since 
$S \subseteq L_C$, it follows that $C \oplus S = (C \cap S^{\bot}) \oplus S 
\oplus S = (C \cap S^{\bot}) \oplus S = C$.
\end{proof}

Finally, given $K \subseteq \N_N$ and defined $\bar{K} = \N_n / K$ and 
\begin{equation*}
 \Lambda(K) = \left\{\lambda \in \R^n: \ \lambda_k > 0 \ \forall k \in K, 
\quad \lambda_k = 0 \ \forall k \in \bar{K} \right\}
\end{equation*}
(note that $\lambda_k$ is strictly positive if and only if $k \in K$) one has 
\begin{equation}\label{eq:lambdaK}
\begin{array}{l}
\displaystyle \{\lambda \in \R^n: \ \lambda_k \geq 0 \ \forall k \in \N_n, \ 
\mathbf{1}^T \lambda = 1\} \\ 
\hspace{1cm} = \bigcup_{K \subseteq \N_N} \{\lambda \in 
\Lambda(K): \mathbf{1}^T \lambda = 1\}
\end{array}
\end{equation}
where $K$ denotes the set of indices such that $\lambda_k$ is not zero, in 
practice. In fact, for every $\lambda$ in the l.h.s. of (\ref{eq:lambdaK}), 
there exists a $K$, that is the set of indices for which $\lambda_k > 0$, such 
that $\lambda \in \Lambda(K)$. Analogously, every $\lambda$ in the r.h.s. of 
(\ref{eq:lambdaK}), also satisfies $\lambda \geq 0$ and then it is contained in 
the l.h.s. set. Equality (\ref{eq:lambdaK}) implies that 
\begin{equation}\label{eq:OinfLambda}
\begin{array}{l}
\check{\Omega}_\infty^{\alpha} = \bigcup_{\substack{\lambda \geq 0 \\ 
\mathbf{1}^T \lambda = 1}} \left( \bigoplus_{k = 1}^N A^{-k} \lambda_k \bO_k 
\right) = \bigcup_{K \subseteq \N_N} \bigcup_{\substack{\lambda \in \Lambda(K)  
\\ \mathbf{1}^T \lambda = 1}} \Bigg( \bigoplus_{k \in K} A^{-k} \lambda_k 
\bO_k\\ 
\hspace{1cm} \oplus \bigoplus_{k \in \bar{K}} A^{-k} \lambda_0 \bO_k \Bigg),\\
\hat{\Omega}_\infty^{\alpha} = \bigcup_{\substack{\lambda \geq 0 
\\ \mathbf{1}^T \lambda = 1}} \left( \bigoplus_{k = 1}^N \lambda_k A^{-k} \bO_k 
\right) = \bigcup_{K \subseteq \N_N} \bigcup_{\substack{\lambda \in \Lambda(K)  
\\ \mathbf{1}^T \lambda = 1}} \Bigg( \bigoplus_{k \in K} \lambda_k A^{-k} 
\bO_k\\ 
\hspace{1cm} \oplus \bigoplus_{k \in \bar{K}} \lambda_0 A^{-k} \bO_k \Bigg),\\
\end{array}
\end{equation}
with $\lambda_0 = 0$.

\begin{lemma}\label{lem:CuDplusE}
Given the sets $C, D, E \subseteq \R^n$, one has $(C \cup D) \oplus E = (C 
\oplus E) \cup (D \oplus E)$.
\end{lemma}

\begin{proof}
From the definition of Minkowski sum, it follows
\begin{equation*}
 \begin{array}{l}
(C \cup D) \oplus E = \{x \in \R^n : \ x \in C \ \mathrm{ 
or } \ x \in D \} \oplus E  =  \{x + y \in \R^n : \\ 
\ (x \in C, \ y \in E) \ \mathrm{ or } \ (x \in D, \ y \in E) \} =  \{x + y \in 
\R^n : \ x \in C, \\ \ y \in E\} \cup \{x + y \in \R^n : x \in D, \ y \in E) \} 
 =  (C \oplus E) \cup (D \oplus E).
 \end{array}
\end{equation*}
\end{proof}

Now we are in the position of proving that $\Omega_\infty^{\alpha} 
= \hat{\Omega}_\infty^{\alpha}$, even for singular $A$.

\begin{theorem}\label{th:theoremOinf}
Let $\Omega$ and $U$ as in (\ref{eq:OmegaU}) be non-empty and such that $0 \in 
\Omega$ and $0 \in U$. Then the sets $\check{\Omega}_\infty^{\alpha}$ and 
$\hat{\Omega}_\infty^{\alpha}$, defined in (\ref{eq:Oinfs}), are equal.
\end{theorem}

\begin{proof}
If $A$ is nonsingular, the equality follows directly from 
Lemma~\ref{lem:gammaA}. Consider now the case of $A$ singular. The sets 
$\lambda_k A^{-k} \bO_k$ and $A^{-k} \lambda_k \bO_k$, involved in 
(\ref{eq:Oinfs}), are equal for every $k \in \N_N$ provided $\lambda_k > 0$, 
from Lemma~\ref{lem:gammaA}. On the other hand, this is no more true if 
$\lambda_k = \lambda_0 = 0$, in fact
\begin{equation}\label{eq:th1}
\begin{array}{ll}
 \lambda_0 A^{-k} \bO_k \hspace{-0.2cm} & = \{\lambda_0 x \in \R^n : \ x \in 
A^{-k} \bO_k \}\\
& = \{\lambda_0 x \in \R^n : \ A^k x \in \bO_k \} = \{0\},\\
A^{-k} \lambda_0 \bO_k \hspace{-0.2cm} & = \{x \in \R^n : \ A^k x \in \lambda_0 
\bO_k\} \\
& = \left\{x \in \R^n : \ A^k x \in \{0\}\right\} = \ker(A^k),
\end{array}
\end{equation}
as the set of $x$ such that $A^k x \in \bO_k$ is non-empty from $0 \in \Omega$ 
and $0 \in U$. Moreover, for every $k \in \N_N$ one has 
\begin{equation}\label{eq:kerA}
\begin{array}{cl}
\ker(A^k) & \hspace{-0.2cm} = \{x \in \R^n : \ A^k x \in \{0\}\}\\
& \subseteq \{x \in \R^n : \ A^k x \in \lambda_k \bO_k\} = A^{-k} \lambda_k 
\bO_k,\\
\end{array}
\end{equation}
since $0 \in \lambda_k \bO_k$ for every $\lambda_k \geq 0$, from $0 \in \Omega$ 
and $0 \in U$. Inclusion (\ref{eq:kerA}) with $\lambda_k = 1$ implies 
$\bigcup_{k = 1}^N \ker(A^k) \subseteq \bigcup_{k = 1}^N A^{-k} \bO_k$ and then
\begin{equation}\label{eq:kerAinObar}
\bigoplus_{k = 1}^N \ker(A^k) = \co \left( \bigcup_{k = 1}^N \ker(A^k) \right) 
\subseteq \co \left( \bigcup_{k = 1}^N A^{-k} \bO_k \right) = 
\hat{\Omega}_\infty^\alpha,
\end{equation}
where the first equality follows from Lemma~\ref{lem:S1S2} and the fact that 
$\ker(A^k)$ are subspaces. Moreover, (\ref{eq:kerA}) yields
\begin{equation}\label{eq:kerAinO}
\bigoplus_{k = 1}^N \ker(A^k) = \hspace{-0.2cm} \bigcup_{\substack{\lambda \geq 
0 \\ \mathbf{1}^T \lambda = 1}} \hspace{-0.2cm} \left( \bigoplus_{k = 1}^N 
\ker(A^k) \right) \hspace{-0.1cm} \subseteq \hspace{-0.2cm} 
\bigcup_{\substack{\lambda \geq 0 \\ \mathbf{1}^T \lambda = 1}} \hspace{-0.2cm} 
\left( \bigoplus_{k = 1}^N A^{-k} \lambda_k \bO_k \right) \hspace{-0.1cm} = 
\Omega_\infty^\alpha.
\end{equation}
Then, denoting the value $\lambda_0 = 0$, one has  
\begin{equation*}\label{eq:theoremOinfZ}
\begin{array}{l}
 \hat{\Omega}_\infty^{\alpha} = \co \left(\bigcup_{k = 1}^{N} 
A^{-k} \bO_k \right) = \co \left(\bigcup_{k = 1}^{N} A^{-k} \bO_k \right) 
\oplus \bigoplus_{k = 1}^N \ker(A^k) \\
 \quad = \bigcup_{\substack{\lambda \geq 0 \\ \mathbf{1}^T \lambda = 1}} \left( 
\bigoplus_{k = 1}^N \lambda_k A^{-k} \bO_k \right) \oplus \bigoplus_{k = 1}^N 
\ker(A^k) \\
 \quad = \bigcup_{K \subseteq \N_N} \bigcup_{\substack{\lambda \in 
\Lambda(K)  \\ \mathbf{1}^T \lambda = 1}} \hspace{-0.1cm} \left( \bigoplus_{k 
\in K} \lambda_k A^{-k} \bO_k \oplus \bigoplus_{k \in \bar{K}} \lambda_0 A^{-k} 
\bO_k \hspace{-0.1cm} \right) \hspace{-0.1cm} \oplus \bigoplus_{k = 1}^N 
\ker(A^k)\\
 \quad = \bigcup_{K \subseteq \N_N} \bigcup_{\substack{\lambda \in 
\Lambda(K)  \\ \mathbf{1}^T \lambda = 1}} \left( \bigoplus_{k \in K} \lambda_k 
A^{-k} \bO_k \right) \oplus \bigoplus_{k = 1}^N \ker(A^k)\\
 \quad = \bigcup_{K \subseteq \N_N} \bigcup_{\substack{\lambda \in 
\Lambda(K)  \\ \mathbf{1}^T \lambda = 1}} \left( \bigoplus_{k \in K} A^{-k} 
\lambda_k \bO_k \right) \oplus \bigoplus_{k = 1}^N \ker(A^k)\\
 \quad = \bigcup_{K \subseteq \N_N} \bigcup_{\substack{\lambda \in 
\Lambda(K)  \\ \mathbf{1}^T \lambda = 1}} \left( \bigoplus_{k \in K} A^{-k} 
\lambda_k \bO_k \oplus \bigoplus_{k = 1}^N \ker(A^k)\right) \\
 \quad = \bigcup_{K \subseteq \N_N} \bigcup_{\substack{\lambda \in 
\Lambda(K)  \\ \mathbf{1}^T \lambda = 1}} \left( \bigoplus_{k \in K} A^{-k} 
\lambda_k \bO_k \oplus \bigoplus_{k = 1}^N \ker(A^k) \oplus \bigoplus_{k \in 
\bar{K}} \ker(A^k)\right) \\
 \quad = \hspace{-0.1cm} \bigcup_{K \subseteq \N_N} \bigcup_{\substack{\lambda 
\in \Lambda(K)  \\ \mathbf{1}^T \lambda = 1}} \hspace{-0.1cm} \left( 
\hspace{-0.1cm} \left( \bigoplus_{k \in K} A^{-k} \lambda_k \bO_k \oplus 
\bigoplus_{k \in \bar{K}} \ker(A^k) \hspace{-0.1cm} \right) \hspace{-0.1cm} 
\oplus \bigoplus_{k = 1}^N \ker(A^k) \hspace{-0.1cm} \right) \\
 \quad = \bigcup_{K \subseteq \N_N} \bigcup_{\substack{\lambda \in 
\Lambda(K)  \\ \mathbf{1}^T \lambda = 1}} \hspace{-0.1cm} \left( \bigoplus_{k 
\in K} A^{-k} \lambda_k \bO_k \oplus \bigoplus_{k \in \bar{K}} A^{-k} \lambda_0 
\bO_k \hspace{-0.1cm} \right) \hspace{-0.1cm} \oplus \bigoplus_{k = 1}^N 
\ker(A^k) \\
 \quad = \bigcup_{\substack{\lambda \geq 0 \\ \mathbf{1}^T \lambda 
= 1}} \left( \bigoplus_{k = 1}^N A^{-k} \lambda_k \bO_k \right) \oplus 
\bigoplus_{k = 1}^N \ker(A^k) \\
\quad = \bigcup_{\substack{\lambda \geq 0 \\ \mathbf{1}^T \lambda = 1}} \left( 
\bigoplus_{k = 1}^N A^{-k} \lambda_k \bO_k \right) = 
\check{\Omega}_\infty^{\alpha} 
\end{array}
\end{equation*}
where: the second equality holds from Lemma~\ref{lem:CSS} and 
(\ref{eq:kerAinObar}); the forth from (\ref{eq:OinfLambda}); the fifth from 
(\ref{eq:th1}); the sixth from Lemma~\ref{lem:gammaA}; the seventh from 
Lemma~\ref{lem:CuDplusE}; the eighth from Lemma~\ref{lem:CSS}; the tenth from 
(\ref{eq:th1}); the eleventh from (\ref{eq:OinfLambda}); the twelfth and the 
last one from Lemma~\ref{lem:CSS} and (\ref{eq:kerAinO}).
\end{proof}

Theorem~\ref{th:theoremOinf} implies that checking if $x \in 
\Omega_{\infty}^{\alpha}$ resorts to solve an LP feasibility problem in the 
variables $x, z_k, v_{i,k}, \lambda_k$ for all $i \in \N_k$ and $k \in \N_N$, 
then in a space of dimension $n + Nn + 0.5 N(N+1) m + N$.

\begin{remark}\label{rem:dimensions}
From Theorem~\ref{th:theoremOinf}, also the stop condition 
(\ref{eq:invariance_cond1}), employed in Algorithm~\ref{alg:I}, can be posed as 
an LP problem, once $\alpha_N$ is fixed. In fact, by reasonings analogous to 
those of Theorem~\ref{th:Ninclusion}, the inclusion $\alpha_N \Omega \subseteq 
\Omega_{\infty}^{\alpha}$ can be posed in form of the LP problem 
(\ref{eq:FinalCondLP}), by appropriately defining the matrices $\bG$, $\bg$, 
$\bH$, $\bh$ from (\ref{eq:Inv_for_Anonsin}). Such an LP problem could be also 
used to approximate the optimal $\alpha_N$, by griding it for instance. This 
would also have the benefit of leading to smaller values of $N$, since the the 
stop condition (\ref{eq:invariance_cond1}) holds if the N-step one 
(\ref{eq:invariance_condN}) is satisfied, but the inverse is not true in 
general. On the other hand, the dimension of such an LP problem might be much 
bigger than for the N-step condition, in fact: $\bar{n}$ would be $(1+N)n + 
0.5N(N+1)m + N$ instead of $n+Nm$ defined for the N-step; $n_{\bg} = 2n + Nn_h + 
0.5N(N+1)n_g + N + 1$ instead of $n_h + Nn_g$ and $n_{\bh} = n_h + 2nN + N(N+1)m 
+ 2N$ instead of $n_h + 2Nm$.
\end{remark}

\subsection{State constraints}\label{sec:state_constr}

Concerning the state constraints, recall first that every smaller multiple of 
$\Omega_{\infty}^{\alpha}$, i.e. every $\sigma \Omega_{\infty}^{\alpha}$ with 
$\sigma \in [0, \, 1]$, is still control invariant, under 
Assumption~\ref{ass:XU} and if $0 \in \Omega$. Then the greatest invariant 
multiple of $\Omega_{\infty}^{\alpha}$ contained in $X$ is given by 
$\Omega_{\infty}^{\alpha, \sigma} = \sigma \Omega_{\infty}^{\alpha}$ with 
\begin{equation*}\label{eq:OmegaalphasigmaZ}
\begin{array}{ll}
\sigma = \hspace{-0.2cm} & \max_{\delta \in [0, \, 1]} \delta\\
& \mathrm{s.t. } \ \ \delta \Omega_{\infty}^{\alpha} \subseteq X,
\end{array}
\end{equation*}
which is equivalent, for $X$ as in (\ref{eq:XU}), to 
\begin{equation}\label{eq:LPOinX}
\sigma = \min\left\{1, \frac{f_1}{\delta_1}, \ldots, 
\frac{f_{n_f}}{\delta_{n_f}} \right\} \quad \mathrm{with } \ 
\begin{array}{ll}
\delta_i = \hspace{-0.2cm} & \max_{x} F_i x\\
& \mathrm{s.t. } \ \ x \in \Omega_{\infty}^{\alpha},
\end{array}
\end{equation}
for every $i \in \N_{n_f}$. Then the scaling factor $\sigma$ can be obtained by 
computing $\delta_i$ solving (\ref{eq:LPOinX}), which are $n_f$ LP problems, 
in both cases of $A$ singular and nonsingular, from 
Theorem~\ref{th:theoremOinf}. In fact, the constraint $x \in 
\Omega_{\infty}^{\alpha}$ is a set of linear constraints in a space of dimension 
$(1+N)n + 0.5 N(N+1) m + N$, see (\ref{eq:Inv_for_Anonsin}) and 
Remark~\ref{rem:dimensions}.

The method presented above does not explicitly take into account the shape of 
$X$ in computing $\Omega_\infty^{\alpha,\sigma}$, and then could lead to some 
conservatism. Alternatively, the state constraints could be considered by 
defining the analogous of the preimage set $\Omega_k^{\alpha}$ as in 
(\ref{eq:Omegaka}). Given the constraint sets $X \subseteq \R^n$ and $U \in 
\R^m$ and $\sigma \in \R_+$, the set defined by
\begin{equation}\label{eq:OmegakasXU}
\begin{array}{l}
\Omega_k^{\sigma}(X,U)  = \{x \in \R^n: \ A^k x + \sum_{i = 1}^{k} A^{i-1}Bu_i 
\in \sigma \Omega,\\
\ \ A^{k-1} x + A^{k-2} B u_k + \ldots + B u_2 \in X, \ \ldots \ A x + Bu_k \in 
X,\\
\ \ x \in X, \ u_i \in U \ \forall i \in \N_k\}
\end{array}
\end{equation}
is the set of initial states $x \in X$ for which a sequence of length $k$ of 
input $u_i \in U$, with $i \in \N_k$, exists such that the generated trajectory 
is maintained in $X$ and ends in $\sigma \Omega$ at time $k$. Hence, the set 
given by 
\begin{equation*}\label{eq:OmegainfasXUZ}
\Omega_{\infty}^{\sigma}(X,U) = \co\left( \bigcup_{k = 1}^N 
\Omega_k^{\sigma}(X,U) \right) 
\end{equation*}
is the set of states that can be steered in $\sigma \Omega$ in $N$ steps at 
most through an admissible trajectory that does not violate the state 
constraints $X$, if $\sigma \Omega \subseteq \Omega_{\infty}^{\sigma}(X,U)$. 
Hence, solving the problem 
\begin{equation}\label{eq:LPOmegainfasXU}
\begin{array}{ll}
\hat{\sigma} = \hspace{-0.2cm} & \displaystyle \max_{\sigma \in \R_+} \sigma\\
& \mathrm{s.t. } \ \ \sigma \Omega \subseteq \Omega_{\infty}^{\sigma}(X,U),
\end{array}
\end{equation}
leads to the control invariant set $\Omega_{\infty}^{\hat{\sigma}}(X,U)$ 
contained in $X$. The problem (\ref{eq:LPOmegainfasXU}) does not yield to a 
convex problem, as for the sufficient condition (\ref{eq:invariance_cond2}) and 
the considerations that follow it. Nevertheless, the following problem 
\begin{equation}\label{eq:LPOmegainfasXUmu}
\begin{array}{ll}
\mu_N = \hspace{-0.2cm} & \displaystyle \min_{\mu \in \R_+} \mu \\
& \mathrm{s.t. } \ \ \Omega \subseteq \Omega_{\infty}^{1}(\mu X,\mu U),
\end{array}
\end{equation}
leads to an LP analogous to (\ref{eq:FinalCondLPbeta}) and equivalent to 
(\ref{eq:LPOmegainfasXU}), with $\mu = \sigma^{-1}$, since it can be proved 
that $\sigma^{-1} \Omega_k^{\sigma}(X,U) = \Omega_k^{1}(\mu X,\mu U)$ for all $k 
\in \N$.

Finally, note that $\Omega_k^{1}(\mu X, \mu U)$ as in (\ref{eq:OmegakasXU}) is 
the projection on $\R^n$ of a polytope on a space of dimension $n + km + 1$. 
Then the constraint $\Omega \subseteq \Omega_{\infty}^{1}(\mu X,\mu U)$ would 
lead to linear constraints analogous to those of Theorem~\ref{th:Ninclusion} 
but with $\bar{n} = (1+N)n + 0.5N(N+1)m + N$, $n_{\bg} = 2n + Nn_h + 
0.5N(N+1)n_g + N + 1$ and $n_{\bh} = n_h + 2nN + N(N+1)m + 2N$, see also 
Remark~\ref{rem:dimensions}. The solution of the lower dimensional optimization 
problem 
\begin{equation}\label{eq:LPOmegaNasXUmu}
\begin{array}{ll}
\mu_N = \hspace{-0.2cm} & \displaystyle \min_{\mu \in \R_+} \mu \\
& \mathrm{s.t. } \ \ \Omega \subseteq \Omega_{N}^{1}(\mu X,\mu U),
\end{array}
\end{equation}
leads to a more conservative control invariant set contained in $X$. Note that 
(\ref{eq:LPOmegaNasXUmu}) yields to an LP problem analogous of the N-step 
condition in absence of state constraints (\ref{eq:FinalCondLPbeta}).

\section{Numerical examples}

The different results presented in this paper are illustrated through numerical 
examples.

\subsection{Example 1}
The first example concerns a system with $n = 2$ and $m = 1$. The main interest 
relies in fact that both the Minkowski sum and the convex hull can be computed 
efficiently in this low dimensional system, using for instance the MPT toolbox 
for managing polytopes, \cite{MPT3}. This would allow us to explicitly compute 
outer approximations of the maximal invariant set and the sets 
$\Omega_N^{\alpha}$ and $\Omega_\infty^{\alpha}$ and then to give a graphical 
illustration of our results in terms of conservatism. 

We consider the systems (\ref{eq:system}) with 
\begin{equation}
 A = \left[\begin{array}{cc}
 1.2 & 1\\
 0 & 1.2
\end{array}\right], \quad B = \left[\begin{array}{c}
0.5\\
0.3
\end{array}\right] 
\end{equation}
and constraints on the input $U = \{u \in \R: \, \|u\| \leq 2\}$.  We consider 
first no constraints in the state. The initial set $\Omega$ has been chosen to 
be the unitary box, i.e. $\Omega = \B^2$. Then the maximal value of $\alpha$ 
such that sets $\Omega^{\alpha}$ and $\Omega_N^{\alpha}$ satisfy 
(\ref{eq:invariance_cond2}) is obtained for different values of $N$, by solving 
(\ref{eq:FinalCondLPbeta}). Given such $\alpha$, the set 
$\Omega_{\infty}^{\alpha}$, defined in (\ref{eq:Omegaconvex}), is a control 
invariant set. To give an intuition of the method results and of the 
conservatism with respect to the maximal control invariant, a sequence of 
non-increasing nested outer approximations of the maximal control invariant set 
is computed, by starting with $\Sigma_0$ big enough (i.e. containing the maximal 
control invariant set) and iterating $\Sigma_k = \Sigma_0 \cap A^{-1}(\Sigma_k 
\oplus(- B U))$, \cite{BlanchiniBook}. The sets $\Sigma_{0}$ and $\Sigma_{10k}$ 
for $k \in \N_5$ are depicted in Figure~\ref{fig:1.5} in thin lines while 
$\Sigma_{60}$ is the white polytope with thick borders. The set 
$\Omega^{\alpha}$ is the dark-gray box and both $\Omega_N^{\alpha}$ and 
$\Omega_\infty^{\alpha}$ are represented in light gray, for $N = 5$. As can be 
noticed, the sets $\Omega_N^{\alpha}$ and $\Omega_\infty^{\alpha}$ are very 
close, where clearly $\Omega_N^{\alpha} \subseteq \Omega_\infty^{\alpha}$. 

\begin{figure}[H]
    \newlength\fheight
    \newlength\fwidth
    \setlength\fheight{3.5cm}
    \setlength\fwidth{8cm}
    \input{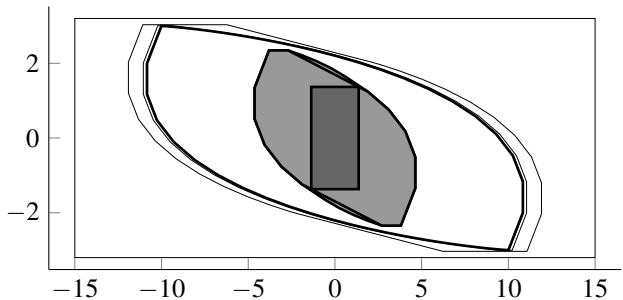}	
    \vspace{-0.35cm}
    \caption{Sets $\Sigma_{0}$ and $\Sigma_{10k}$ for $k \in \N_5$ in thin 
lines; $\Sigma_{60}$ in white with thick lines; $\Omega^{\alpha}$ in dark gray; 
$\Omega_N^{\alpha}$ and $\Omega_\infty^{\alpha}$ in light gray, for $N = 5$.}
	\label{fig:1.5}
\end{figure}

The sets $\Sigma_{60}$ and $\Omega_{\infty}^{\alpha}$ for $N = 5, 10, 15, 20$ 
are drawn in Figure~\ref{fig:1.5-20}, in white the former and gray the latter.

\begin{figure}[H]
    \setlength\fheight{3.5cm}
    \setlength\fwidth{8cm}
    \input{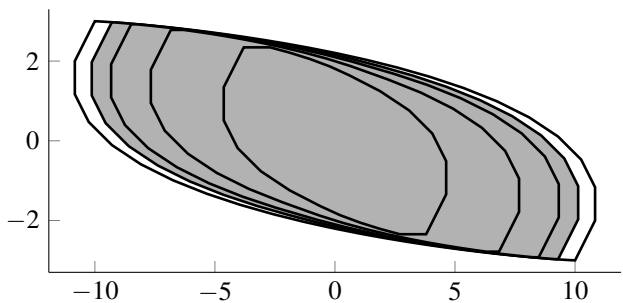}	
    \vspace{-0.35cm}
    \caption{Set $\Sigma_{60}$ in white with thick lines and 
$\Omega_\infty^{\alpha}$ in light gray, for $N = 5$ (inner), $N = 10, 15$ and 
$N = 20$ (outer).}
	\label{fig:1.5-20}
\end{figure}

Finally, the sets $\Sigma_{0}$ and $\Sigma_{10k}$ for $k \in \N_6$ are depicted 
in Figure~\ref{fig:1.5}, in white, together with $\Omega^{\alpha}$, in dark 
gray, and $\Omega_{\infty}^{\alpha}$, in light gray, for $N = 40$. The set 
$\Omega_{\infty}^{\alpha}$ is very close to the outer approximation of the 
maximal control invariant set $\Sigma_{60}$.

\begin{figure}[H]
    \setlength\fheight{3.5cm}
    \setlength\fwidth{8cm}
    \input{CinvN40.tikz}	
    \vspace{-0.35cm}
    \caption{Sets $\Sigma_{0}$ and $\Sigma_{10k}$ for $k \in \N_5$ in thin 
lines; $\Sigma_{60}$ in white with thick lines; $\Omega^{\alpha}$ in dark gray; 
$\Omega_N^{\alpha}$ and $\Omega_\infty^{\alpha}$ in light gray, for $N = 40$.}
	\label{fig:2}
\end{figure}

\subsection{Example 2: singular matrix}

Consider the systems (\ref{eq:system}) with singular transition matrix 
\begin{equation}
 A = \left[\begin{array}{cc}
 1.2 & 1\\
 0 & 0
\end{array}\right], \quad B = \left[\begin{array}{c}
0.5\\
0.3
\end{array}\right] 
\end{equation}
and input constraints sets and initial set as for Example 1, $U = \{u \in \R: \, 
\|u\| \leq 2\}$ and $\Omega = \B^n$. The sets $\Sigma_{i}$ for $i \in \N_{60}$ 
have been computed starting with $\Sigma_0 = 1000 \B^2$. Figure~\ref{fig:Ex2_1} 
shows the outer approximations of the maximal control invariant set $\Sigma_{i}$ 
with $i = 40, 50, 60$, the control invariant sets $\Omega_\infty^{\alpha}$ for 
different values of $N$, in particular $N \in \N_{10}$, and $\Omega^{\alpha}$ 
related to $N = 10$. 
\begin{figure}[H]
    \setlength\fheight{3.5cm}
    \setlength\fwidth{8cm}
%
%
\begin{tikzpicture}

\begin{axis}[%
width=0.951\fwidth,
height=\fheight,
at={(0\fwidth,0\fheight)},
scale only axis,
xmin=-15,
xmax=15,
ymin=-15,
ymax=15,
axis background/.style={fill=white},
axis x line*=bottom,
axis y line*=left
]

\addplot[area legend, line width=0.1pt, draw=black, fill=white, forget plot]
table[row sep=crcr] {%
x	y\\
-1000	-1000\\
1000	-1000\\
1000	1000\\
-1000	1000\\
}--cycle;

\addplot[area legend, line width=0.1pt, draw=black, fill=white, forget plot]
table[row sep=crcr] {%
x	y\\
665.619795106607	-1000\\
1000	-1000\\
1000	-998.743754127926\\
-665.619795106606	1000\\
-1000	1000\\
-1000	998.743754127928\\
}--cycle;

\addplot[area legend, line width=0.1pt, draw=black, fill=white, forget plot]
table[row sep=crcr] {%
x	y\\
799.957952455182	-1000\\
866.708714211486	-1000\\
-799.957952455181	1000\\
-866.708714211485	1000\\
}--cycle;

\addplot[area legend, line width=0.1pt, draw=black, fill=white, forget plot]
table[row sep=crcr] {%
x	y\\
821.654314862112	-1000\\
845.012351804557	-1000\\
-821.65431486211	1000\\
-845.012351804555	1000\\
}--cycle;

\addplot[area legend, line width=0.1pt, draw=black, fill=white, forget plot]
table[row sep=crcr] {%
x	y\\
825.158398519231	-1000\\
841.508268147436	-1000\\
-825.158398519231	1000\\
-841.508268147435	1000\\
}--cycle;

\addplot[area legend, line width=0.1pt, draw=black, fill=white, forget plot]
table[row sep=crcr] {%
x	y\\
825.72432759277	-1000\\
840.942339073898	-1000\\
-825.724327592769	1000\\
-840.942339073897	1000\\
}--cycle;

\addplot[area legend, line width=1.0pt, draw=black, fill=white, forget plot]
table[row sep=crcr] {%
x	y\\
825.72432759277	-1000\\
840.942339073898	-1000\\
-825.724327592769	1000\\
-840.942339073897	1000\\
}--cycle;

\addplot[area legend, line width=1.0pt, draw=black, fill=white!60!black, forget plot]
table[row sep=crcr] {%
x	y\\
826.525663435569	-1000\\
840.141003231098	-1000\\
-826.525663435569	1000\\
-840.141003231098	1000\\
}--cycle;

\addplot[area legend, line width=1.0pt, draw=black, fill=white!60!black, forget plot]
table[row sep=crcr] {%
x	y\\
826.664129456016	-1000\\
840.00253721065	-1000\\
-826.664129456017	1000\\
-840.002537210651	1000\\
}--cycle;

\addplot[area legend, line width=1.0pt, draw=black, fill=white!60!black, forget plot]
table[row sep=crcr] {%
x	y\\
826.830288680553	-1000\\
839.836377986114	-1000\\
-826.830288680553	1000\\
-839.836377986114	1000\\
}--cycle;

\addplot[area legend, line width=1.0pt, draw=black, fill=white!60!black, forget plot]
table[row sep=crcr] {%
x	y\\
827.029679749997	-1000\\
839.63698691667	-1000\\
-827.029679749997	1000\\
-839.63698691667	1000\\
}--cycle;

\addplot[area legend, line width=1.0pt, draw=black, fill=white!60!black, forget plot]
table[row sep=crcr] {%
x	y\\
827.26894903333	-1000\\
839.397717633337	-1000\\
-827.26894903333	1000\\
-839.397717633337	1000\\
}--cycle;

\addplot[area legend, line width=1.0pt, draw=black, fill=white!60!black, forget plot]
table[row sep=crcr] {%
x	y\\
827.556072173329	-1000\\
839.110594493338	-1000\\
-827.556072173329	1000\\
-839.110594493338	1000\\
}--cycle;

\addplot[area legend, line width=1.0pt, draw=black, fill=white!60!black, forget plot]
table[row sep=crcr] {%
x	y\\
827.900619941328	-1000\\
838.766046725339	-1000\\
-827.900619941328	1000\\
-838.766046725339	1000\\
}--cycle;

\addplot[area legend, line width=1.0pt, draw=black, fill=white!60!black, forget plot]
table[row sep=crcr] {%
x	y\\
828.314077262927	-1000\\
838.35258940374	-1000\\
-828.314077262927	1000\\
-838.35258940374	1000\\
}--cycle;

\addplot[area legend, line width=1.0pt, draw=black, fill=white!60!black, forget plot]
table[row sep=crcr] {%
x	y\\
828.810226048845	-1000\\
837.856440617821	-1000\\
-828.810226048846	1000\\
-837.856440617822	1000\\
}--cycle;

\addplot[area legend, line width=1.0pt, draw=black, fill=white!60!black, forget plot]
table[row sep=crcr] {%
x	y\\
-837.261062074719	1000\\
-829.405604591948	1000\\
837.261062074719	-1000\\
829.405604591948	-1000\\
}--cycle;

\addplot[area legend, line width=1.0pt, draw=black, fill=white!60!black, forget plot]
table[row sep=crcr] {%
x	y\\
3.71327448966266	-3.71327448966266\\
3.71327448966266	3.71327448966266\\
-3.71327448966266	3.71327448966266\\
-3.71327448966266	-3.71327448966266\\
}--cycle;

\addplot[area legend, line width=1.0pt, draw=black, fill=white!40!black, forget plot]
table[row sep=crcr] {%
x	y\\
3.71327448966266	-3.71327448966266\\
3.71327448966266	3.71327448966266\\
-3.71327448966266	3.71327448966266\\
-3.71327448966266	-3.71327448966266\\
}--cycle;
\end{axis}
\end{tikzpicture}%
    \vspace{-0.35cm}
    \caption{Sets $\Sigma_{40}$ and $\Sigma_{50}$ in thin lines; $\Sigma_{60}$ 
in white with thick lines; $\Omega_\infty^{\alpha}$ in light gray, for $N \in 
\N_{10}$, and $\Omega^{\alpha}$ for $N = 10$ in dark gray.}
	\label{fig:Ex2_1}
\end{figure}
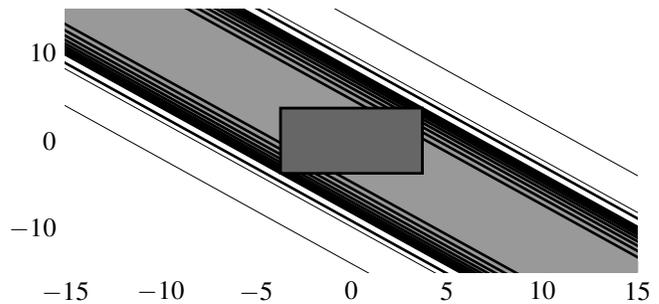

The inner and outer approximation of the maximal invariant appear to be rather 
close for $N = 10$.

\subsection{Example 3: state constraints}

In this example we consider the same dynamics and same sets $\Omega$ and $U$ of 
Example 1 and the state constraint set given by $X = \{x \in \R^2: \ -10 \leq 
x_1 \leq 5, \ -1 \leq x_2 \leq 2\}$. Both method for taking into account the 
state constraints illustrated in Section~\ref{sec:state_constr} are applied 
using $N = 15$. Figure~\ref{fig:Ex3_1} shows the set 
$\Omega_{\infty}^{\alpha,\sigma}$ obtained by solving (\ref{eq:LPOinX}) 
in middle shade gray and also $\Omega_{\infty}^{\sigma}$ induced by the solution 
to (\ref{eq:LPOmegaNasXUmu}) in light gray (besides the sets $X$, $\Sigma_k$ 
and $\sigma 
\Omega$). 

\begin{figure}[H]
    \setlength\fheight{3.5cm}
    \setlength\fwidth{8cm}
%
%
\begin{tikzpicture}

\begin{axis}[%
width=0.951\fwidth,
height=\fheight,
at={(0\fwidth,0\fheight)},
scale only axis,
xmin=-10.75,
xmax=5.75,
ymin=-1.15000000206324,
ymax=2.15000000009825,
axis background/.style={fill=white},
axis x line*=bottom,
axis y line*=left
]

\addplot[area legend, line width=0.1pt, draw=black, fill=white, forget plot]
table[row sep=crcr] {%
x	y\\
-10	-1\\
5	-1\\
5	2\\
-10	2\\
}--cycle;

\addplot[area legend, line width=0.1pt, draw=black, fill=white, forget plot]
table[row sep=crcr] {%
x	y\\
-7.89946305407726	0.47222222222222\\
-7.41720996765751	-0.106481481481479\\
-6.61345482362458	-0.588734567901236\\
-5.60876089358341	-0.990612139917704\\
-5.57746802664243	-1\\
3.60250930481641	-1\\
4.15806486037197	-0.333333333333333\\
4.15806486037197	0.222222222222222\\
3.77226239123617	0.685185185185184\\
3.12925827600983	1.07098765432099\\
2.3255031319769	1.39248971193416\\
1.43244186082921	1.66040809327846\\
0.502169703383693	1.88367341106539\\
-0.0794632412893667	2\\
-7.20501860963282	2\\
-7.89946305407726	1.16666666666667\\
}--cycle;

\addplot[area legend, line width=0.1pt, draw=black, fill=white, forget plot]
table[row sep=crcr] {%
x	y\\
-7.44805795545974	0.472222222222219\\
-6.96580486903999	-0.106481481481477\\
-6.16204972500705	-0.588734567901239\\
-5.1573557949659	-0.990612139917701\\
-5.1260629280249	-1\\
3.37680675550765	-1\\
3.9323623110632	-0.333333333333334\\
3.9323623110632	0.222222222222223\\
3.5465598419274	0.685185185185184\\
2.90355572670107	1.07098765432099\\
2.09980058266814	1.39248971193416\\
1.20673931152045	1.66040809327846\\
0.276467154074933	1.88367341106539\\
-0.305165790598127	2\\
-6.75361351101529	2\\
-7.44805795545974	1.16666666666666\\
}--cycle;

\addplot[area legend, line width=0.1pt, draw=black, fill=white, forget plot]
table[row sep=crcr] {%
x	y\\
-7.37515351188807	0.472222222222227\\
-6.89290042546831	-0.106481481481483\\
-6.0891452814354	-0.588734567901228\\
-5.08445135139421	-0.990612139917706\\
-5.05315848445323	-1\\
3.34035453372181	-1\\
3.89591008927737	-0.333333333333332\\
3.89591008927737	0.222222222222221\\
3.51010762014156	0.685185185185188\\
2.86710350491523	1.07098765432099\\
2.0633483608823	1.39248971193416\\
1.17028708973461	1.66040809327846\\
0.240014932289108	1.88367341106539\\
-0.341618012383967	2\\
-6.68070906744362	2\\
-7.37515351188807	1.16666666666667\\
}--cycle;

\addplot[area legend, line width=0.1pt, draw=black, fill=white, forget plot]
table[row sep=crcr] {%
x	y\\
-7.36337903723377	0.472222222222219\\
-6.88112595081402	-0.106481481481477\\
-6.07737080678109	-0.588734567901233\\
-5.07267687673991	-0.990612139917707\\
-5.04138400979893	-1\\
3.33446729639466	-1\\
3.89002285195022	-0.333333333333333\\
3.89002285195022	0.222222222222223\\
3.50422038281442	0.685185185185184\\
2.86121626758808	1.07098765432099\\
2.05746112355515	1.39248971193416\\
1.16439985240747	1.66040809327846\\
0.234127694961945	1.88367341106539\\
-0.347505249711116	2\\
-6.66893459278932	2\\
-7.36337903723376	1.16666666666666\\
}--cycle;

\addplot[area legend, line width=0.1pt, draw=black, fill=white, forget plot]
table[row sep=crcr] {%
x	y\\
-7.3614773938415	0.472222222222228\\
-6.87922430742175	-0.106481481481481\\
-6.07546916338882	-0.588734567901236\\
-5.07077523334767	-0.990612139917698\\
-5.03948236640666	-1\\
3.33351647469853	-1\\
3.88907203025409	-0.333333333333333\\
3.88907203025408	0.222222222222222\\
3.50326956111828	0.685185185185184\\
2.86026544589194	1.07098765432099\\
2.05651030185901	1.39248971193416\\
1.16344903071133	1.66040809327846\\
0.233176873265815	1.88367341106539\\
-0.348456071407245	2\\
-6.66703294939706	2\\
-7.3614773938415	1.16666666666666\\
}--cycle;

\addplot[area legend, line width=1.0pt, draw=black, fill=white, forget plot]
table[row sep=crcr] {%
x	y\\
-7.3614773938415	0.472222222222228\\
-6.87922430742175	-0.106481481481481\\
-6.07546916338882	-0.588734567901236\\
-5.07077523334767	-0.990612139917698\\
-5.03948236640666	-1\\
3.33351647469853	-1\\
3.88907203025409	-0.333333333333333\\
3.88907203025408	0.222222222222222\\
3.50326956111828	0.685185185185184\\
2.86026544589194	1.07098765432099\\
2.05651030185901	1.39248971193416\\
1.16344903071133	1.66040809327846\\
0.233176873265815	1.88367341106539\\
-0.348456071407245	2\\
-6.66703294939706	2\\
-7.3614773938415	1.16666666666666\\
}--cycle;

\addplot[area legend, line width=1.0pt, draw=black, fill=white!80!black, forget plot]
table[row sep=crcr] {%
x	y\\
-6.79686621202518	0.472222222222222\\
-6.31461312560542	-0.106481481481483\\
-5.51085798157251	-0.588734567901229\\
-4.50616405153132	-0.990612139917704\\
-4.47487118459034	-1\\
3.18188723311969	-1\\
3.73744278867524	-0.333333333333332\\
3.73744278867524	0.222222222222221\\
3.35164031953944	0.685185185185185\\
2.7086362043131	1.07098765432099\\
1.90488106028018	1.39248971193416\\
1.01181978913247	1.66040809327847\\
0.0815476316869729	1.88367341106539\\
-0.500085312986087	2\\
-6.10242176758073	2\\
-6.79686621202518	1.16666666666667\\
}--cycle;

\addplot[area legend, line width=1.0pt, draw=black, fill=white!80!black, forget plot]
table[row sep=crcr] {%
x	y\\
-6.79686621202518	0.472222222222222\\
-6.31461312560542	-0.106481481481483\\
-5.51085798157251	-0.588734567901229\\
-4.50616405153132	-0.990612139917704\\
-4.47487118459034	-1\\
3.18188723311969	-1\\
3.73744278867524	-0.333333333333332\\
3.73744278867524	0.222222222222221\\
3.35164031953944	0.685185185185185\\
2.7086362043131	1.07098765432099\\
1.90488106028018	1.39248971193416\\
1.01181978913247	1.66040809327847\\
0.0815476316869729	1.88367341106539\\
-0.500085312986087	2\\
-6.10242176758073	2\\
-6.79686621202518	1.16666666666667\\
}--cycle;

\addplot[area legend, line width=1.0pt, draw=black, fill=white!80!black, forget plot]
table[row sep=crcr] {%
x	y\\
-6.79686621202518	0.472222222222222\\
-6.31461312560542	-0.106481481481483\\
-5.51085798157251	-0.588734567901229\\
-4.50616405153132	-0.990612139917704\\
-4.47487118459034	-1\\
3.18188723311969	-1\\
3.73744278867524	-0.333333333333332\\
3.73744278867524	0.222222222222221\\
3.35164031953944	0.685185185185185\\
2.7086362043131	1.07098765432099\\
1.90488106028018	1.39248971193416\\
1.01181978913247	1.66040809327847\\
0.0815476316869729	1.88367341106539\\
-0.500085312986087	2\\
-6.10242176758073	2\\
-6.79686621202518	1.16666666666667\\
}--cycle;

\addplot[area legend, line width=1.0pt, draw=black, fill=white!60!black, forget plot]
table[row sep=crcr] {%
x	y\\
-2.66251777457199	-0.0632329911035556\\
-2.24987456750052	-0.228290273932145\\
-1.791382115199	-0.365838009622598\\
-1.31378581071813	-0.480461122698009\\
-0.836189506237313	-0.575980383594171\\
-0.227895570922772	-0.682847862190093\\
0.141677436165759	-0.740171275995069\\
0.489462868326317	-0.787940787499215\\
0.812457333671697	-0.827748713752715\\
1.10926378135688	-0.860921985630531\\
1.37963948212199	-0.888566378862168\\
1.62415006139222	-0.911603373221822\\
1.84390679019927	-0.930800868521427\\
2.04036899149849	-0.946798781271257\\
2.21519715421179	-0.960130375229356\\
2.3701453411667	-0.971240036861156\\
2.50698387309364	-0.980498088220852\\
2.62744518536987	-0.988213131020855\\
2.73318728087533	-0.994642333354038\\
2.82577041763847	-1.00000000196499\\
2.90548209489566	-1.00000000196499\\
3.19070107962343	-0.657737220291663\\
3.19070107962343	-0.37251823556389\\
2.99263234022913	-0.134835748290725\\
2.66251777457201	0.0632329911035473\\
2.2498745675005	0.228290273932151\\
1.79138211519902	0.365838009622593\\
1.31378581071814	0.480461122698005\\
0.836189506237296	0.575980383594174\\
0.227895570922905	0.68284786219007\\
-0.141677436165856	0.740171275995082\\
-0.489462868326363	0.78794078749922\\
-0.812457333671466	0.827748713752686\\
-1.10926378135743	0.86092198563059\\
-1.37963948212187	0.888566378862158\\
-1.62415006139146	0.911603373221753\\
-1.84390679020009	0.930800868521494\\
-2.04036899149837	0.946798781271248\\
-2.21519715421191	0.960130375229366\\
-2.37014534116662	0.971240036861151\\
-2.50698387309443	0.980498088220907\\
-2.62744518536935	0.988213131020826\\
-2.73318728087414	0.994642333353969\\
-2.82577041763858	1.000000001965\\
-2.90548209489566	1.000000001965\\
-3.19070107962344	0.657737220291658\\
-3.19070107962344	0.372518235563898\\
-2.99263234022913	0.134835748290732\\
}--cycle;

\addplot[area legend, line width=1.0pt, draw=black, fill=white!40!black, forget plot]
table[row sep=crcr] {%
x	y\\
0.999999999954711	-0.999999999954711\\
0.999999999954711	0.999999999954711\\
-0.999999999954711	0.999999999954711\\
-0.999999999954711	-0.999999999954711\\
}--cycle;
\end{axis}
\end{tikzpicture}%
    \vspace{-0.35cm}
    \caption{Sets $\Sigma_{0} = X$ and $\Sigma_{10k}$ with $k \in \N_5$ in thin 
lines; $\Sigma_{60}$ in white with thick lines; $\Omega_\infty^{\sigma}$ in 
light gray; $\Omega_\infty^{\alpha,\sigma}$ in middle shade gray, and $\sigma 
\Omega$ in dark gray, for $N = 15$.}
	\label{fig:Ex3_1}
\end{figure}
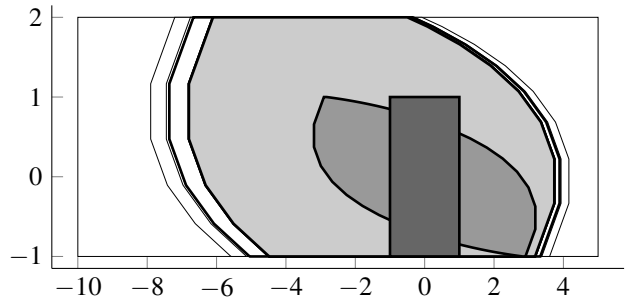

Note how the conservatism with respect to the scaling procedure 
(\ref{eq:LPOinX}) is reduced by taking into account the shape of $X$ as 
in the method based on (\ref{eq:LPOmegaNasXUmu}). The latter, in fact, provides 
a good approximation of the maximal control invariant set for $N = 15$.

\subsection{Example 4: high dimensional system}

We apply now the proposed method to an high dimensional system, in particular 
with $n = 20$ and $m = 10$. To provide some hints on the conservatism of the 
control invariant obtained with respect to the maximal control invariant set, we 
build a system for which the latter can be computed, or, at least, approximated. 
Indeed the classical algorithms for computing or approximating the maximal 
control invariant set are too computationally demanding to be applied to high 
dimensional systems in general. Then, a particular structure has to be imposed 
to the system dynamics to apply them and obtain an estimation of the maximal 
control invariant set to be compared with our results. In particular, we 
consider system (\ref{eq:system}) with 
\begin{equation*}
 A = \left[\begin{array}{cccc}
 A_1 & 0 & \ldots & 0\\
 0& A_2 & \ldots & 0\\
 \ldots & \ldots & \ldots & \ldots \\
 0 & 0 & \ldots & A_{10}\\
 \end{array}\right] \!\!\!, \quad  
 B = \left[\begin{array}{ccccc}
 B_1 & 0 & \ldots & 0\\
 0& B_2 & \ldots & 0\\
 \ldots & \ldots & \ldots & \ldots \\
 0 & 0 & \ldots & B_{10}\\
 \end{array}\right]
\end{equation*}
where $A_i \in \R^{2 \times 2}$ and $B_i \in \R^{2}$, for $i \in \N_{10}$, are 
matrices whose entries are randomly generated such that all $A_i$ have instable 
poles and the pairs $(A_i, B_i)$ are controllable. This means that the whole 
system is controllable and it is, in practice, composed by $10$ decoupled 
two-dimensional subsystems with one control input each. Hence, the maximal 
control invariant set for the overall system, $\Sigma_{\infty}$, is given by the 
Cartesian product of the maximal control invariant sets of the 10 subsystems, 
that is $\textstyle \Sigma_{\infty} = \prod_{i = 1}^{10}\Sigma_{i,\infty}$ 
where $\Sigma_{i,\infty}$ are the maximal control invariant set (or an outer 
approximation of it) for the $i$-th subsystem. Then $\Sigma_{\infty}$ can be 
computed by computing $\Sigma_{i,\infty}$, being $(A_i, B_i)$ a two-dimensional 
controllable system, for all $i \in \N_{10}$. 

The linear problem (\ref{eq:FinalCondLPbeta}) has been posed with $N = 3, 5, 
9, 15$ and solved with YALMIP interface \cite{Lofberg2004} and Mosek optimizer 
\cite{mosek}. In Table~\ref{tb:1}, the dimensions and solution times for the LP 
problems are reported.

\begin{table}[H]
\begin{center}
\begin{tabular}{|c|c|c|c|c|}  
  \hline
   & $N = 3$ & $N = 5$ & $N = 9$ & $N = 15$ \\
  \hline
  LP dimension  & 10002 & 19602 & 48402 & 115602 \\
  Solution time & $0.9s$ & $0.99s$ & $1.25s$ & $1.71s$\\
  \hline
\end{tabular}
\vspace*{0.1cm}
\caption{\label{tb:1}} 
\end{center}
\end{table}

To quantify the difference between the outer approximation of the maximal 
control invariant set $\Sigma_{\infty}$ and the set $\Omega_{\infty}^{\alpha}$, 
$100$ vectors $v \in \R^n$ are generated randomly. Then, (a lower approximation 
of) the maximal values of $r_{\Sigma}$ and $r_{\Omega}$ are computed such that 
$r_{\Sigma} v \in \Sigma_{\infty}$ and $r_{\Omega} v \in 
\Omega_{\infty}^{\alpha}$, through dichotomy method. In practice, we search for 
(approximations of) the intersections between the ray $v_r = \{r v \in \R^n: \ 
r \geq 0\}$ and the boundaries of the sets $\Sigma_{\infty}$ and 
$\Omega_{\infty}^{\alpha}$. The ratio between $r_{\Omega} / r_{\Sigma}$ is an 
indicator of the mismatch between the outer approximation of the maximal 
control invariant set $\Sigma_{\infty}$ and $\Omega_{\infty}^{\alpha}$, the 
closer to one, the closer are the intersections between the ray $v_r$ and the 
two sets.

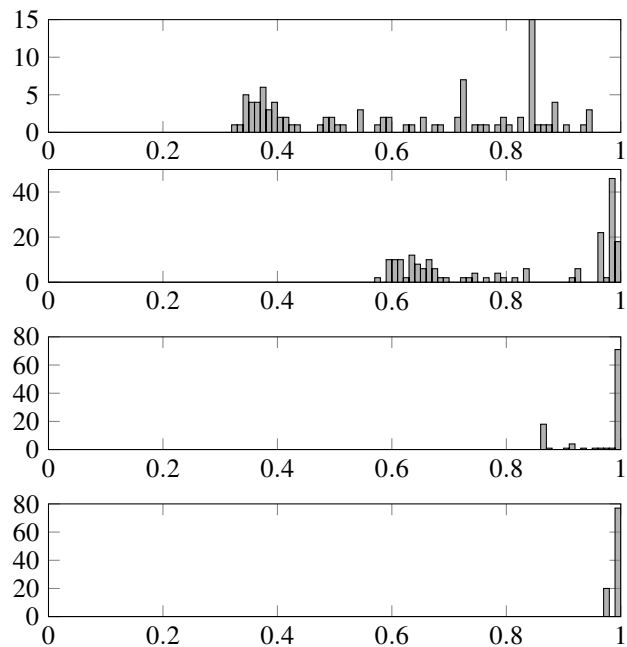
\begin{figure}[H]
    \setlength\fheight{1.5cm}
    \setlength\fwidth{8cm}
%
%
\begin{tikzpicture}

\begin{axis}[%
width=0.951\fwidth,
height=\fheight,
at={(0\fwidth,0\fheight)},
scale only axis,
xmin=0,
xmax=1,
ymin=0,
ymax=15,
axis background/.style={fill=white}
]
\addplot[ybar interval, fill=gray, fill opacity=0.6, draw=black, area legend] table[row sep=crcr] {%
x	y\\
0	0\\
0.01	0\\
0.02	0\\
0.03	0\\
0.04	0\\
0.05	0\\
0.06	0\\
0.07	0\\
0.08	0\\
0.09	0\\
0.1	0\\
0.11	0\\
0.12	0\\
0.13	0\\
0.14	0\\
0.15	0\\
0.16	0\\
0.17	0\\
0.18	0\\
0.19	0\\
0.2	0\\
0.21	0\\
0.22	0\\
0.23	0\\
0.24	0\\
0.25	0\\
0.26	0\\
0.27	0\\
0.28	0\\
0.29	0\\
0.3	0\\
0.31	0\\
0.32	1\\
0.33	1\\
0.34	5\\
0.35	4\\
0.36	4\\
0.37	6\\
0.38	3\\
0.39	4\\
0.4	2\\
0.41	2\\
0.42	1\\
0.43	1\\
0.44	0\\
0.45	0\\
0.46	0\\
0.47	1\\
0.48	2\\
0.49	2\\
0.5	1\\
0.51	1\\
0.52	0\\
0.53	0\\
0.54	3\\
0.55	0\\
0.56	0\\
0.57	1\\
0.58	2\\
0.59	2\\
0.6	0\\
0.61	0\\
0.62	1\\
0.63	1\\
0.64	0\\
0.65	2\\
0.66	0\\
0.67	1\\
0.68	1\\
0.69	0\\
0.7	0\\
0.71	2\\
0.72	7\\
0.73	0\\
0.74	1\\
0.75	1\\
0.76	1\\
0.77	0\\
0.78	1\\
0.79	2\\
0.8	1\\
0.81	0\\
0.82	2\\
0.83	0\\
0.84	15\\
0.85	1\\
0.86	1\\
0.87	1\\
0.88	4\\
0.89	0\\
0.9	1\\
0.91	0\\
0.92	0\\
0.93	1\\
0.94	3\\
0.95	0\\
0.96	0\\
0.97	0\\
0.98	0\\
0.99	0\\
1	0\\
};
\end{axis}
\end{tikzpicture}%
%
%
\definecolor{mycolor1}{rgb}{0.00000,0.44700,0.74100}%
\begin{tikzpicture}

\begin{axis}[%
width=0.951\fwidth,
height=\fheight,
at={(0\fwidth,0\fheight)},
scale only axis,
xmin=0,
xmax=1,
ymin=0,
ymax=50,
axis background/.style={fill=white}
]
\addplot[ybar interval, fill=gray, fill opacity=0.6, draw=black, area legend] 
table[row sep=crcr] {%
x	y\\
0	0\\
0.01	0\\
0.02	0\\
0.03	0\\
0.04	0\\
0.05	0\\
0.06	0\\
0.07	0\\
0.08	0\\
0.09	0\\
0.1	0\\
0.11	0\\
0.12	0\\
0.13	0\\
0.14	0\\
0.15	0\\
0.16	0\\
0.17	0\\
0.18	0\\
0.19	0\\
0.2	0\\
0.21	0\\
0.22	0\\
0.23	0\\
0.24	0\\
0.25	0\\
0.26	0\\
0.27	0\\
0.28	0\\
0.29	0\\
0.3	0\\
0.31	0\\
0.32	0\\
0.33	0\\
0.34	0\\
0.35	0\\
0.36	0\\
0.37	0\\
0.38	0\\
0.39	0\\
0.4	0\\
0.41	0\\
0.42	0\\
0.43	0\\
0.44	0\\
0.45	0\\
0.46	0\\
0.47	0\\
0.48	0\\
0.49	0\\
0.5	0\\
0.51	0\\
0.52	0\\
0.53	0\\
0.54	0\\
0.55	0\\
0.56	0\\
0.57	2\\
0.58	0\\
0.59	10\\
0.6	10\\
0.61	10\\
0.62	2\\
0.63	12\\
0.64	8\\
0.65	6\\
0.66	10\\
0.67	6\\
0.68	2\\
0.69	2\\
0.7	0\\
0.71	0\\
0.72	2\\
0.73	2\\
0.74	4\\
0.75	0\\
0.76	2\\
0.77	0\\
0.78	4\\
0.79	2\\
0.8	0\\
0.81	2\\
0.82	0\\
0.83	6\\
0.84	0\\
0.85	0\\
0.86	0\\
0.87	0\\
0.88	0\\
0.89	0\\
0.9	0\\
0.91	2\\
0.92	6\\
0.93	0\\
0.94	0\\
0.95	0\\
0.96	22\\
0.97	2\\
0.98	46\\
0.99	18\\
1	18\\
};
\end{axis}
\end{tikzpicture}%
%
%
\begin{tikzpicture}

\begin{axis}[%
width=0.951\fwidth,
height=\fheight,
at={(0\fwidth,0\fheight)},
scale only axis,
xmin=0,
xmax=1,
ymin=0,
ymax=80,
axis background/.style={fill=white}
]
\addplot[ybar interval, fill=gray, fill opacity=0.6, draw=black, area legend] table[row sep=crcr] {%
x	y\\
0	0\\
0.01	0\\
0.02	0\\
0.03	0\\
0.04	0\\
0.05	0\\
0.06	0\\
0.07	0\\
0.08	0\\
0.09	0\\
0.1	0\\
0.11	0\\
0.12	0\\
0.13	0\\
0.14	0\\
0.15	0\\
0.16	0\\
0.17	0\\
0.18	0\\
0.19	0\\
0.2	0\\
0.21	0\\
0.22	0\\
0.23	0\\
0.24	0\\
0.25	0\\
0.26	0\\
0.27	0\\
0.28	0\\
0.29	0\\
0.3	0\\
0.31	0\\
0.32	0\\
0.33	0\\
0.34	0\\
0.35	0\\
0.36	0\\
0.37	0\\
0.38	0\\
0.39	0\\
0.4	0\\
0.41	0\\
0.42	0\\
0.43	0\\
0.44	0\\
0.45	0\\
0.46	0\\
0.47	0\\
0.48	0\\
0.49	0\\
0.5	0\\
0.51	0\\
0.52	0\\
0.53	0\\
0.54	0\\
0.55	0\\
0.56	0\\
0.57	0\\
0.58	0\\
0.59	0\\
0.6	0\\
0.61	0\\
0.62	0\\
0.63	0\\
0.64	0\\
0.65	0\\
0.66	0\\
0.67	0\\
0.68	0\\
0.69	0\\
0.7	0\\
0.71	0\\
0.72	0\\
0.73	0\\
0.74	0\\
0.75	0\\
0.76	0\\
0.77	0\\
0.78	0\\
0.79	0\\
0.8	0\\
0.81	0\\
0.82	0\\
0.83	0\\
0.84	0\\
0.85	0\\
0.86	18\\
0.87	1\\
0.88	0\\
0.89	0\\
0.9	1\\
0.91	4\\
0.92	0\\
0.93	1\\
0.94	0\\
0.95	1\\
0.96	1\\
0.97	1\\
0.98	1\\
0.99	71\\
1	71\\
};
\end{axis}
\end{tikzpicture}%
%
%
\begin{tikzpicture}

\begin{axis}[%
width=0.951\fwidth,
height=\fheight,
at={(0\fwidth,0\fheight)},
scale only axis,
xmin=0,
xmax=1,
ymin=0,
ymax=80,
axis background/.style={fill=white}
]
\addplot[ybar interval, fill=gray, fill opacity=0.6, draw=black, area legend] table[row sep=crcr] {%
x	y\\
0	0\\
0.01	0\\
0.02	0\\
0.03	0\\
0.04	0\\
0.05	0\\
0.06	0\\
0.07	0\\
0.08	0\\
0.09	0\\
0.1	0\\
0.11	0\\
0.12	0\\
0.13	0\\
0.14	0\\
0.15	0\\
0.16	0\\
0.17	0\\
0.18	0\\
0.19	0\\
0.2	0\\
0.21	0\\
0.22	0\\
0.23	0\\
0.24	0\\
0.25	0\\
0.26	0\\
0.27	0\\
0.28	0\\
0.29	0\\
0.3	0\\
0.31	0\\
0.32	0\\
0.33	0\\
0.34	0\\
0.35	0\\
0.36	0\\
0.37	0\\
0.38	0\\
0.39	0\\
0.4	0\\
0.41	0\\
0.42	0\\
0.43	0\\
0.44	0\\
0.45	0\\
0.46	0\\
0.47	0\\
0.48	0\\
0.49	0\\
0.5	0\\
0.51	0\\
0.52	0\\
0.53	0\\
0.54	0\\
0.55	0\\
0.56	0\\
0.57	0\\
0.58	0\\
0.59	0\\
0.6	0\\
0.61	0\\
0.62	0\\
0.63	0\\
0.64	0\\
0.65	0\\
0.66	0\\
0.67	0\\
0.68	0\\
0.69	0\\
0.7	0\\
0.71	0\\
0.72	0\\
0.73	0\\
0.74	0\\
0.75	0\\
0.76	0\\
0.77	0\\
0.78	0\\
0.79	0\\
0.8	0\\
0.81	0\\
0.82	0\\
0.83	0\\
0.84	0\\
0.85	0\\
0.86	0\\
0.87	0\\
0.88	0\\
0.89	0\\
0.9	0\\
0.91	0\\
0.92	0\\
0.93	0\\
0.94	0\\
0.95	0\\
0.96	0\\
0.97	20\\
0.98	0\\
0.99	77\\
1	77\\
};
\end{axis}
\end{tikzpicture}%
    \vspace{-0.35cm}
    \caption{Histograms of the values $r_{\Omega} / r_{\Sigma}$ for $N = 3, 5, 
9, 15$ (from top to bottom).}\label{fig:Ex4_1}
\end{figure}

Figure~\ref{fig:Ex4_1} shows the histograms of the ratio $r_{\Omega} / 
r_{\Sigma}$ for the different values of $N$. As expected, the higher is the 
horizon $N$, the closer are the sets $\Sigma_{\infty}$ and 
$\Omega_{\infty}^{\alpha}$.

\section{Conclusions}

In this paper we addressed the problem of computing control invariant sets for 
linear systems with state and input polyhedral constraints. In particular we 
considered the computational complexity inherent to the explicit determination 
of polyhedral one-step sets, that are the basis of many iterative procedures 
for obtaining control invariant sets. Invariance conditions are given, that are 
set inclusions involving the N-step sets, which are posed in form of LP 
optimization problems, instead of Minkowski sum of polyhedra. Then the 
procedures based on those conditions  are applicable even for high dimensional 
systems. 

\bibliography{biblio}
\bibliographystyle{plain}

\balance

\end{document}